\documentclass[%
 preprint,
 showpacs,
 showkeys,
 amsmath,
 amssymb,
 aps,
 pra,
  longbibliography,
 ]{revtex4-1}
 \usepackage{graphicx}
\usepackage[dvipsnames]{xcolor}
\usepackage{tikz}

\usepackage{amsthm}
\newtheorem{Theorem}{Theorem}

\newtheorem{Corollary}[Theorem]{Corollary}

\newtheorem{Assumption}{Assumption}
\theoremstyle{definition}

\newtheorem{Example}[Theorem]{Example}

\newcommand{\ket}[1]{\left| #1 \right>}

\newcommand{\iprod}[2]{\langle #1 | #2 \rangle}

\newcommand{\oprod}[2]{| #1 \rangle\langle #2 |}

\newcommand{\seq}[1]{\mathbf{#1}}

\newcommand{\m}[1]{\widetilde{#1}}

\begin{document}

\sloppy

\title{Strong Kochen-Specker Theorem and Incomputability of Quantum Randomness}

\author{Alastair A. Abbott}
\email{a.abbott@auckland.ac.nz}
\homepage{http://www.cs.auckland.ac.nz/~aabb009}

\author{Cristian S. Calude}
\email{cristian@cs.auckland.ac.nz}
\homepage{http://www.cs.auckland.ac.nz/~cristian}

\author{Jonathan Conder}
\email{jonathan.conder@auckland.ac.nz}

\affiliation{Department of Computer Science, University of Auckland,
Private Bag 92019, Auckland, New Zealand}

\author{Karl Svozil}
\affiliation{Institute for Theoretical Physics,
Vienna  University of Technology,
Wiedner Hauptstrasse 8-10/136,
1040 Vienna ,  Austria}
\email{svozil@tuwien.ac.at}
\homepage{http://tph.tuwien.ac.at/~svozil}

\date{\today}

\begin{abstract}
The Kochen-Specker theorem shows the impossibility for a hidden variable theory to consistently assign values to certain (finite) sets of observables in a way that is non-contextual and consistent with quantum mechanics.
If we require non-contextuality, the consequence is that many observables must not have pre-existing definite values.
However, the Kochen-Specker theorem does not allow one to determine \emph{which} observables must be value indefinite.
In this paper we present an improvement on the Kochen-Specker theorem which allows one to actually locate observables which are \emph{provably value indefinite}.

Various technical and subtle aspects relating to this formal proof and its connection to quantum mechanics are discussed.
This result is then utilized for the proposal and certification of a dichotomic quantum random number generator operating in a three-dimensional Hilbert space.
\end{abstract}

\preprint{CDMTCS preprint nr. 422/2012}
\maketitle

\section{Located quantum value indefiniteness}

While Bell's theorem \cite{bell-66} expresses the impossibly for a local hidden variable theory to
give the same statistical results as quantum mechanics, the Kochen-Specker theorem \cite{specker-60,kochen1}
proves the impossibility for a hidden variable theory to even assign values to certain (finite) sets of observables in a way that is non-contextual and consistent with quantum mechanics.
More precisely, it expresses a contradiction between the following presuppositions:
\begin{itemize}
        \item[(P1)]
        the set of observables in question \footnote{Which, due to complementarity, may not be all simultaneously co-measurable (i.e., formally, commuting)} have pre-assigned definite values,
        \item[(P2)]
        the outcomes of measurements of observables are non-contextual; that is, they are independent of whatever other
co-measurable observables are measured alongside them,
\end{itemize}
along with the requirement that the relationship between hidden variables associated with
sets of co-measurable observables behave quasi-classically, as expected from quantum mechanics.
This requirement means that in any ``complete'' set of mutually co-measurable yes-no propositions
(represented by mutually orthogonal projectors spanning the Hilbert space) exactly one proposition should be assigned the value ``yes.''

Thereby, the Kochen-Specker theorem does {\em not explicitly identify}
certain particular observables which violate one or more of these presuppositions.
Indeed, the Kochen-Specker theorem has not been designed to actually
{\em locate} the particular observable(s) which would violate the assumptions.
This is not seen as a deficiency of the theorem,
because its content suffices for the many (mostly metaphysical)
purposes it has been designed for and applied to.

In what follows we shall pursue a threefold agenda.
First, we shall make explicit and formalize the physical notions involved,
in particular, value (in-)definiteness and contextuality.
We shall thereby remain within the formalism of quantum logic,
as outlined by Birkhoff and von Neumann \cite{v-neumann-55,birkhoff-36},
as well as by Kochen and Specker \cite{kochen2,kochen3}.

This enables us to specify exactly the actual
{\em location of breakdown of classicality} within the set of Kochen-Specker observables;
that is, we  identify the observables for which classicality
inadvertently renders complete contradictions,
no matter what their (classical) outcome or value may be.
In order to do this, we prove a modified version of the original
Kochen-Specker theorem in which we obtain a contradiction between
the presupposition (P2) and a crucially weaker version of (P1).

Second, we will clarify in what sense the Kochen-Specker and Bell-type theorems imply
the violation of the non-contextuality assumption (P2).
Formalization has become necessary because in the literature the term ``contextuality''
is often identified with violations of certain Bell-type inequalities
on single quanta \cite{hasegawa:230401,Bartosik-09,kirch-09,PhysRevLett.103.160405}
in the absence of strict locality conditions \cite{cabello:210401}.

We point out that, while from a purely logical point of view,  violation of the non-contextuality assumption (P2)
is {\em sufficient} to interpret the Kochen-Specker theorem,
it is by no means {\em necessary} for, or implied by,
the Kochen-Specker theorem.
Indeed, violation of the primary assumption of value definiteness (P1) presents a viable (albeit also not necessary, as other, more
exotic, possibilities demonstrate; e.g., Ref. \cite{pitowsky-82}) option to interpret the Kochen-Specker theorem.

Third, we shall also consider which  collections of observables do not render Kochen-Specker contradictions.
Restricting ourselves to these very limited collections would allow maintenance of assumptions (P1) and (P2)
about quantized systems, but would also reduce the domain of conceivable observables dramatically.

The results presented can be interpreted as one natural consequence of,
and advancement beyond, the Kochen-Specker theorem.
They may be particularly important if we  investigate the concrete ``underpinning'' of the Kochen-Specker theorem:
exactly why and where a quantized system disobeys   classicality.

Apart from foundational issues,
there is also a concrete application which profits from such quantum information theoretic findings.
Contemporary quantum random number generators can no longer be based upon and certified by
our conviction in the quantum postulate of {\em complementarity} alone.
They should also be certified by
strictly stronger forms of non-classicality than complementarity,
{\em quantum value indefiniteness} being one of them \footnote{Note that there exist models of complementarity such as automaton logic or generalised urn models which are value definite \cite{svozil-2001-eua}.}.
For these purposes, the Kochen-Specker theorem, as well as other Bell-type theorems,
serve merely as {\em indications} that  quantum value indefiniteness possibly ``happens somewhere'' because it cannot be excluded that particular individual quanta \footnote{In the Bell-type cases all observables,
and in the Kochen-Specker case ``many'' observables.} could still be value definite.

Unfortunately, by their very design, these theorems cannot guarantee that a particular observable actually {\em is} value indefinite.
One could, for instance, not exclude that a ``demon'' could act in such a way that all observables actually measured would be value definite,
whereas other observables which are not measured would be value indefinite.

However, for quantum random number generators we need certification of value indefiniteness on the
{\em particular observables utilized for that purpose}.
Thus, one needs a different (in the sense of locatedness of violation of non-classicality, stronger) type of theorem
than Kochen and Specker present, an argument that could (formally) {\em assure} that, if quantum mechanics is correct,
the particular quantum observables used for the generation of random number sequences are {\em provably value indefinite,} hence  the measured quantum sequences cannot refer to any consistent property of the measured quanta alone.

This article presents such an argument, which will be utilized for a dichotomic quantum random number generator operating
in a three-dimensional Hilbert space.
By now it should be clear that such a device would be strictly preferential to previous proposals
using merely quantum complementarity, or, in addition to that, some type of non-located violations of global value definiteness.

In what follows we shall first present the basic definitions, then state and prove the aforementioned result,
and subsequently apply this result to
the proposal of a quantum random number generator based on {\em located quantum value indefiniteness} which produces, as we prove, a strongly incomputable sequence of bits.

\section{Definitions}

\subsection{Notation and formal framework}

As usual we denote the set of complex numbers by $\mathbb{C}$
and use the standard quantum mechanical bra-ket notation;
that is, we denote vectors in the Hilbert space $\mathbb{C}^n$ by $\ket{\cdot}$.
We will have particular interest in the projection operators 
 projecting on to the linear subspace spanned by a non-zero vector $\ket{\psi}$, namely $P_\psi = \frac{\oprod{\psi}{\psi}}{\iprod{\psi}{\psi}}\raisebox{.3mm}{;}$
we will use this notation for projection operators throughout this paper.
We briefly note that in this paper we only consider pure quantum states,
and will accordingly not explicitly specify quantum states as pure states as opposed to mixed states.



In order to discuss hidden variable theories precisely and without any of the ambiguity that is common in such discussion,
we present an explicit formal framework in which we will work.

We fix a positive integer $n$.
Let $\mathcal{O} \subseteq \{ P_\psi \mid \ket{\psi} \in \mathbb{C}^n \}$ be a nonempty set of projection observables in the Hilbert space $\mathbb{C}^n$ and $\mathcal{C}\subseteq \{\{P_1,P_2,\dots P_n\} \mid P_i \in \mathcal{O} \text{ and } \iprod{i}{j}=0 \text{ for } i\neq j \}$ a set of measurement contexts over $\mathcal{O}$.
A context $C\in\mathcal{C}$ is thus a maximal set of compatible (i.e. they can be simultaneous measured) projection observables.
Let $v : \{(o,C) \mid o\in \mathcal{O}, C\in\mathcal{C}\text{ and } o\in C \} \xrightarrow{o} \{0,1\}$ be a partial function
(i.e., it may be undefined for some values in its domain).
For some $o,o'\in \mathcal{O}$ and $C,C'\in\mathcal{C}$ we say $v(o,C)=v(o',C')$ if $v(o,C),v(o',C')$ are both defined and have equal values.
If either $v(o,C)$ or $v(o',C')$ are not defined or they are both defined but have different values, then $v(o,C)\neq v(o',C')$.
We will call $v$ an \emph{assignment function}, and it expresses the notion of a hidden variable: it specifies in advance the result obtained from the measurement of an observable.




An observable $o\in C$ is \emph{value definite} in the context $C$ under $v$ if $v(o,C)$ is defined.
Otherwise $o$ is \emph{value indefinite} in $C$.
If $o$ is value definite in all contexts $C\in \mathcal{C}$ for which $o\in C$ then we simply say that $o$ is value definite under $v$.
Similarly, if $o$ is value indefinite in all such contexts $C$ then we say that $o$ is value indefinite under $v$.
The set $\mathcal{O}$ is \emph{value definite} under $v$ if every observable $o\in\mathcal{O}$ is value definite under $v$.
This notion of value definiteness corresponds to the classical notion of determinism:
an observable is value definite if $v$ assigns it a definite value---i.e.\ we are able to predict in advance the value obtained via measurement.

An observable $o\in \mathcal{O}$ is \emph{non-contextual} under $v$ if for all contexts $C,C' \in \mathcal{C}$ with $o\in C,C'$ we have $v(o,C)=v(o,C')$. Otherwise, $v$ is \emph{contextual}.
Note that an observable which is value indefinite in a context is always contextual even if it takes the same value in every context in which it is value definite.
On the other hand, if an observable is value definite in all contexts that it is in, it can be either contextual or not (and in the latter case its value is constant in all contexts containing it) depending on $v$.
The set of observables $\mathcal{O}$ is \emph{non-contextual} under $v$ if every observable $o\in\mathcal{O}$ which is not value indefinite (i.e.\ value definite in \emph{some} context) is non-contextual under $v$.
Otherwise, the set of observables $\mathcal{O}$ is \emph{contextual}.
Further, we say that the set of observables $\mathcal{O}$ is \emph{strongly contextual} under $v$ if every observable $o\in\mathcal{O}$ is contextual under $v$.
Non-contextuality corresponds to the classical notion that the value obtained via measurement is independent of other compatible observables measured alongside it.

Every strongly contextual set of observables under $v$ is contextual under $v$, provided that $v$ is not undefined everywhere.
However the converse implication is false, as
we will discuss in Sec~\ref{sec:strongContextuality}.

If an observable $o$ is non-contextual then it is value definite, but this is not true for sets of observables: $\mathcal{O}$ can be non-contextual but not value definite if it contains an observable which is value indefinite.

An assignment function $v$ is \emph{admissible} if the following hold for all $C \in \mathcal{C}$:
\begin{itemize}
        \item if there exists an $o\in C$ with $v(o,C)=1$, then $v(o',C)=0$ for all $o'\in C\setminus\{o\}$,
        \item if there exists an $o\in C$ such that $v(o',C)=0$ for all $o'\in C\setminus \{o\}$, then $v(o,C)=1$.
\end{itemize}

In the discussion of hidden variables, we do not concern ourselves with the mechanism of $v$, but rather with its possible existence subject to certain constraints (specifically, the admissibility of $v$---we justify this more fully in Sec~\ref{sec:interpretation}---requires that functions of the values associated with compatible observables satisfy the predictions of quantum theory).
The notion of admissibility serves as an analog to the notion of a \emph{two-valued (dispersionless) measure} that is used in quantum logic
\cite{ZirlSchl-65,kalmbach-86,Alda,Alda2,pulmannova-91,svozil-tkadlec},
the difference being that the definition is sound even when not all observables are value definite.
This distinction is subtle but, nevertheless, will allow us to formulate known results, such as the Kochen-Specker theorem
\cite{kochen1},
 as well as the stronger results which we present in this paper.
However, we stress that this is still a purely formal framework and that, in order to make a connection to physical reality, further assumptions must be made, specifically pertaining to the nature of measurement;
we defer this connection to physical reality to Sec~\ref{sec:interpretation}.

We briefly note that this formal framework could be presented in an even more more abstract setting without reference to Hilbert spaces, but for the sake of concreteness we avoid this here.
\if01
\begin{Example}
        As an example to aid familiarity with our definitions, let $n=3$ and consider the set of observables $\mathcal{O}=\{0,1,2,3,4,5,6,7,8\}$ and contexts $\mathcal{C}=\{C_1=\{0,1,2\},C_2=\{0,3,4\},C_3=\{0,5,6\},C_4=\{6,7,8\}\}$.
This configuration is depicted by a Greechie\footnote{Observables are represented by circles, contexts by smooth line segments.}
orthogonality diagram \cite{greechie:71,pulmannova-91,svozil-tkadlec} in Fig.~\ref{2012-incomputability-proofs-f1}.
        Let our assignment function be defined as $$v(0,C_1)=v(0,C_2)=v(6,C_4)=1,$$ $$v(1,C_1)=v(2,C_1)=v(3,C_2)=v(4,C_2)=v(6,C_3)=v(7,C_4)=v(8,C_4)=0,$$ and undefined elsewhere.
        Observables 0 and 6 are both contextual: although $v(0,C_1)=v(0,C_2)=1,$ observable 0 is value indefinite in $C_3$ since $v(0,C_3)$ is not defined; observable 6 is value definite but we have $v(6,C_3)\neq v(6,C_4)$.
        Observable 5 is value indefinite, since it appears only in $C_3$ and $v(5,C_3)$ is not defined.
        The other observables only appear in one context, in which they are all defined, and are thus non-contextual.
        This set $\mathcal{O}$ is neither value definite nor non-contextual.
        The function $v$ is admissible, but the function $v'$ specified exactly as $v$, except that it is defined for observable 5 in $C_3$ as $v'(5,C_3)=0$, would not be admissible since the second condition for admissibility would not be satisfied in $C_3$.
\end{Example}

\begin{figure}[htbp]
\begin{center}
\begin{tikzpicture}
\tikzstyle{every path}=[line width=1pt]
\tikzstyle{every node}=[draw,line width=1pt,inner sep=0]

\tikzstyle{c1}=[circle,minimum size=6]
\tikzstyle{c2}=[circle,minimum size=12]
\tikzstyle{c3}=[circle,minimum size=18]

\tikzstyle{l1}=[draw=none,circle,minimum size=9]
\tikzstyle{l2}=[draw=none,circle,minimum size=15]
\tikzstyle{l3}=[draw=none,circle,minimum size=21]

\draw[Green] (0:0) -- (150:3)
        coordinate[c3,at start] (0)
        coordinate[c1,midway] (1)
        coordinate[c1,at end] (2);

\draw[Orange] (0.center) -- (30:3)
        coordinate[c2,at start] (0)
        coordinate[c1,midway] (3)
        coordinate[c1,at end] (4);

\draw[Red] (0.center) -- (270:3)
        coordinate[c1,at start] (0)
        coordinate[c1,midway] (5)
        coordinate[c2,at end] (6);

\draw[Blue] (6.center) -- (3,-3)
        coordinate[c1,at start] (6)
        coordinate[c1,midway] (7)
        coordinate[c1,at end] (8);

\coordinate[l3,label=90:0] (0) at (0.center);
\coordinate[l1,label=60:1] (1) at (1.center);
\coordinate[l1,label=60:2] (2) at (2.center);

\coordinate[l1,label=120:3] (3) at (3.center);
\coordinate[l1,label=120:4] (4) at (4.center);

\coordinate[l1,label=180:5] (5) at (5.center);
\coordinate[l2,label=180:6] (6) at (6.center);

\coordinate[l1,label=270:7] (7) at (7.center);
\coordinate[l1,label=270:8] (8) at (8.center);
\end{tikzpicture}
\end{center}
\caption{(Color online)
Greechie orthogonality diagram associated with the configuration of observables
 $\mathcal{C}=\{C_1=\{0,1,2\},C_2=\{0,3,4\},C_3=\{0,5,6\},C_4=\{6,7,8\}\}$.
Different contexts $C_i$ are drawn in different colours.
}
\label{2012-incomputability-proofs-f1}
\end{figure}
\fi

\subsection{Kochen-Specker theorem}

Using the framework developed, the Kochen-Specker theorem~\cite{kochen1}, which we outlined and discussed in the introduction, can be presented in the following more rigorous form:
if $n>2$ there exists a set of projection observables $\mathcal{O}$ on $\mathbb{C}^n$ and a set of contexts over $\mathcal{O}$ such that there is no admissible assignment function $v$ under which $\mathcal{O}$ is both non-contextual and value definite.
This proves that it is impossible for all projection observables to be value definite and non-contextual.



\subsection{Strong contextuality can not be guaranteed}
\label{sec:strongContextuality}

How strong is the incompatibility between non-contextuality and value definiteness stated in the Kochen-Specker theorem?
The theorem tells us that not every observable can be both non-contextual and value definite, but gives us no information regarding how far this incompatibility goes.
Here we show that this incompatibility cannot be maximal:
no set of observables is strongly contextual under every admissible value definite assignment function on it.
In other words, for any set of contexts over any set of observables, there exists an admissible assignment function under which the set of observables is value definite and at least one observable is non-contextual.\\
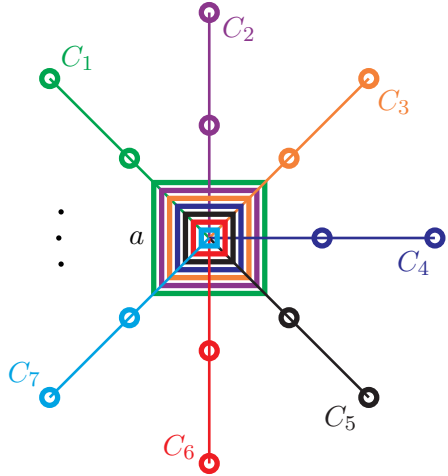
\begin{figure}[htbp] 
\begin{center}
\begin{tikzpicture}
\tikzstyle{every path}=[line width=1pt]
\tikzstyle{every node}=[draw,line width=2pt,inner sep=0]

\tikzstyle{c1}=[rectangle,minimum size=6]
\tikzstyle{c2}=[rectangle,minimum size=12]
\tikzstyle{c3}=[rectangle,minimum size=18]
\tikzstyle{c4}=[rectangle,minimum size=24]
\tikzstyle{c5}=[rectangle,minimum size=30]
\tikzstyle{c6}=[rectangle,minimum size=36]
\tikzstyle{c7}=[rectangle,minimum size=42]

\tikzstyle{d1}=[circle,draw=none,fill,minimum size=2]

\tikzstyle{l7}=[draw=none,circle,minimum size=45]

\draw[Green] (0:0) -- (135:3)
        coordinate[c7,at start] (0)
        coordinate[c1,circle,midway] (1)
        coordinate[c1,circle,at end,label=35:$C_1$] (2);

\draw[Fuchsia] (0.center) -- (90:3)
        coordinate[c6,at start] (0)
        coordinate[c1,circle,midway] (3)
        coordinate[c1,circle,at end,label=350:$C_2$] (4);

\draw[Orange] (0.center) -- (45:3)
        coordinate[c5,at start] (0)
        coordinate[c1,circle,midway] (5)
        coordinate[c1,circle,at end,label=305:$C_3$] (6);

\draw[Blue] (0.center) -- (0:3)
        coordinate[c4,at start] (0)
        coordinate[c1,circle,midway] (7)
        coordinate[c1,circle,at end,label=260:$C_4$] (8);

\draw[Black] (0.center) -- (315:3)
        coordinate[c3,at start] (0)
        coordinate[c1,circle,midway] (9)
        coordinate[c1,circle,at end,label=215:$C_5$] (10);

\draw[Red] (0.center) -- (270:3)
        coordinate[c2,at start] (0)
        coordinate[c1,circle,midway] (11)
        coordinate[c1,circle,at end,label=170:$C_6$] (12);

\draw[Cerulean] (0.center) -- (225:3)
        coordinate[c1,at start] (0)
        coordinate[c1,circle,midway] (13)
        coordinate[c1,circle,at end,label=125:$C_7$] (14);

\coordinate[l7,label=180:$a$] (0) at (0.center);

\coordinate[d1] (.) at (190:2);
\coordinate[d1] (.) at (180:2);
\coordinate[d1] (.) at (170:2);
\end{tikzpicture}
\end{center}
\caption{(Color online)
Greechie orthogonality diagram\footnote{Observables are represented by circles and squares, contexts by smooth line segments.}
of the contexts in $S_a$ with an overlaid value assignment reflecting the
argument against strong-contextuality being guaranteed.
Different contexts $C_i$ are drawn in different colours, circles represent the value 0 and squares represent the value 1.
}
\label{2012-incomputability-proofs-f2}
\end{figure}

More precisely, let $\mathcal{O}$ be a set of projection observables and $\mathcal{C}$ a set of contexts over $\mathcal{O}$.
Then for every $a\in \mathcal{O}$ there exists an admissible assignment function $v$ such that $v(a,C)=1$ for every context $C\in\mathcal{C}$ with $o \in C$, and $\mathcal{O}$ is value definite under $v$.
To see this, consider the set $S_{a}=\{C \mid C\in \mathcal{C} \text{ and } a \in C \}\subseteq \mathcal{C}$ of contexts in which $a$ appears.
If we define the assignment function $v_{a}$ for $C\in S_{a}$ by

\begin{equation*}
        v_{a}(o,C) =
        \begin{cases}
                1, & \text{for $o=a$,}\\
                0, & \text{for  $o\neq a$.}\\
        \end{cases}
\end{equation*}

It is clear this satisfies $\sum_{o\in C}v_{a}(o,C)=1$, for all $C\in S_{a}$.
For $C\in \mathcal{C}\setminus S_{a}$, the function $v_{a}$ can be defined in any arbitrary contextual way to satisfy
admissibility.
The function $v_{a}$ is then admissible and assigns a definite value (namely 1) to the observable $a$ (which was arbitrarily chosen) in a non-contextual way---i.e.\ $v_a(a, C) = 1$ for all $C \in S_{a}$.

Note that the configuration of contexts   $S_{a}=\{C \mid C\in \mathcal{C} \text{ and } a \in C \}\subseteq \mathcal{C}$
amounts to a ``star-shaped'' Greechie orthogonality diagram, with the common observable $a$ at the center of the star,
as depicted in Fig.~\ref{2012-incomputability-proofs-f2}.

Indeed this should not be surprising in view of the predictions of quantum mechanics.
Specifically, for a physical system prepared in the state $\ket{\psi}$, the Born rule predicts that measurement of the projection observable $P_\psi$ should give the value 1 (non-contextually) with probability 1.
Nevertheless, it is important to place a bound on the degree of non-classicality~\cite{Elitzur199225,svozil-2011-enough} that we can guarantee.
In fact, it is possible to go further than we have and define $v_a$ to non-contextually assign the value 0 to each observable appearing on a ``ray'' of the star in Fig.~\ref{2012-incomputability-proofs-f2}.
This is a consequence of the fact no two observables on differing ``rays'' are compatible.

However, in the following we show that one cannot go much further than this.
Specifically, in what are the main theoretical results of the paper, we show that there are pairs of observables (belonging to different contexts) such that at most one of them can be assigned the value 1 by an admissible assignment function under which $\mathcal{O}$ is non-contextual.
This finding is somewhat stronger than a similar result by Kochen and Specker \cite{kochen1,pulmannova-91}
derived from the (as Specker used to call them~\cite{Specker-priv}) ``bug''-type orthogonality diagrams (a sub-diagram of their diagram $\Gamma_1$), as not all observables are assumed to be value definite.
Instead, an observable is only deduced to be value definite where the admissibility of $v$ requires it to be so.

This difference allows us to deduce an even stronger result, with particular relevance to quantum random number generators: there are pairs of observables such that, if one of them is assigned the value 1 by an admissible assignment function under which $\mathcal{O}$ is non-contextual, the other must be {\em value indefinite}.
This is the best guarantee of located value indefiniteness one could hope for, and we will make use of it in our proposal for a quantum random number generator.
The proof relies on the weaker result described above, so we demonstrate that first, and deduce the main result as a corollary.
Note that there are larger values than $\frac{3}{\sqrt{14}}$ for which these results are true.
However, this number is more than sufficient for our purposes, and the larger values we found require significantly longer proofs.\\

\begin{Theorem}
\label{thm:twonotvaluedefinite}
        Let $\ket{a}, \ket{b}\in \mathbb{C}^3$ be unit vectors such that $0 < |\iprod{a}{b}| \le \frac{3}{\sqrt{14}}\raisebox{.8mm}{.}$
        Then there exists a set of projection observables $\mathcal{O}$ containing $P_a$ and $P_b$, and a set of contexts $\mathcal{C}$ over $\mathcal{O}$, such that there is no admissible assignment function under which $\mathcal{O}$ is non-contextual and $P_a$, $P_b$ have the value 1.
\end{Theorem}
\begin{proof}
        We first show that the theorem holds under the equality $|\iprod{a}{b}|=\frac{3}{\sqrt{14}}\raisebox{.8mm}{,}$ and then, by means of a reduction to the case of equality, show it also holds for $0<|\iprod{a}{b}|<\frac{3}{\sqrt{14}}\raisebox{.8mm}{.}$

        By choosing the basis appropriately, without loss of generality we may assume that $\ket{a}\equiv(1,0,0)$ and $\ket{b}\equiv\frac{1}{\sqrt{14}}(3,2,1)$.
        Let $\ket{\psi}=(0,1,0)$ and $\ket{\phi}=(0,0,1)$.

        In Table~\ref{2012-incomputability-proofs-table:greechie} we define 24 contexts $C_1, C_2, \dots, C_{24}$, which are numbered by the column headings.
        Each row vector $\ket{\varphi}$ in the table is defined relative to the afore-chosen basis $\{ \ket{a}, \ket{\psi}, \ket{\phi} \}$, and is understood to represent the corresponding projection observable $P_\varphi$.
        For brevity, we have omitted commas, brackets and normalisation constants from these vectors, and have used the notation $\m{n} = -n$.
        As an example, $C_1 = \{ P_a, P_\psi, P_\phi \}$.

        \begin{table}
                \caption{Assignment table containing the representation of observable propositions (projectors), together with the contexts in which they appear.
                         See Fig.~\ref{2012-incomputability-proofs-fig:greechie} for an illustration of these.}
                \label{2012-incomputability-proofs-table:greechie}
                \small
                $$\begin{array}{|c|c|c|c|c|c|c|c|c|c|c|c|c|c|}
                        \hline
                        v & 1, 2        & 3, 4          & 5, 6      & 7, 8       & 9, 10      & 11, 12      & 13        & 14, 15      & 16, 17      & 18, 19      & 20, 21      & 22, 23      & 24      \\
                        \hline
                        1 & 1~0~0       & 1~0~0         & \bf 2~1~1 & 2~1~1        & \bf 2~0~1   & 2~0~1       &           & 1~1~0       & \bf 1~1~1   & 1~1~1       & \bf 1~0~1   & 1~0~1       &         \\
                        0 & \bf 0~1~0   & \bf 0~1~1     & 1~\m1~\m1 & \bf 1~0~\m2  & 1~0~\m2     & \bf 1~1~\m2 &           & \bf 1~\m1~0 & 1~\m1~0     & \bf 1~0~\m1 & 1~0~\m1     & \bf 1~1~\m1 &         \\
                        0 & \bf 0~0~1   & \bf 0~1~\m1   & 0~1~\m1   & 2~\m5~1      & 0~1~0       & 1~\m5~\m2   &           & 0~0~1       & 1~1~\m2     & 1~\m2~1     & 0~1~0       & 1~\m2~\m1   & 1~1~\m1 \\
                          &             &               &           &              &             &             &           &             &             &             &             &             & 1~\m1~0 \\
                        1 & 3~2~1       & 3~2~1         & \bf 3~2~0 & 3~2~0        & \bf 3~1~\m1 & 3~1~\m1     & \bf 1~1~0 & 1~1~0       & \bf 2~1~\m1 & 2~1~\m1     & \bf 2~0~\m1 & 2~0~\m1     & 1~1~2   \\
                        0 & \bf 2~\m3~0 & \bf 1~\m1~\m1 & 2~\m3~0   & \bf 2~\m3~3  & 2~\m3~3     & \bf 1~\m1~2 & 1~\m1~2   & \bf 1~\m1~1 & 1~\m1~1     & \bf 1~0~2   & 1~0~2       & \bf 1~1~2   &         \\
                        0 & 3~2~\m{13}  & 1~\m4~5       & 0~0~1     & 6~\m9~\m{13} & 0~1~1       & 1~\m7~\m4   & 1~\m1~\m1 & 1~\m1~\m2   & 0~1~1       & 2~\m5~\m1   & 0~1~0       & 1~\m5~2     &         \\
                        \hline
                \end{array}$$
                \normalsize
        \end{table}

        Now let $\mathcal{C} = \{ C_1, C_2, \dots, C_{24} \}$ and $\mathcal{O} = \bigcup_{i = 1}^{24} C_i$.
        Suppose there exists an admissible assignment function $v$ under which $\mathcal{O}$ is non-contextual and $v(P_a, C_1) = v(P_b, C_2) = 1$.
        By continual application of the admissibility requirements, one can show that $v$ assigns certain values to all the observables in Table~\ref{2012-incomputability-proofs-table:greechie}.
        This argument proceeds through the table from left to right, where the value assigned to each observable is noted in the leftmost column.
        For example, in the first step we conclude that $v(P_\psi, C_1) = v(P_\phi, C_1) = 0$.
        An observable whose value is determined by the others in the column is marked in bold, provided that the value given will be used later on.
        This argument is also illustrated in Fig.~\ref{2012-incomputability-proofs-fig:greechie}.
        We eventually obtain a contradiction, namely that $v(o, C_{24}) = 0$ for all $o \in C_{24}$ (the dotted line in Fig.~\ref{2012-incomputability-proofs-fig:greechie}).
        Therefore there does not exist such admissible assignment function $v$.

        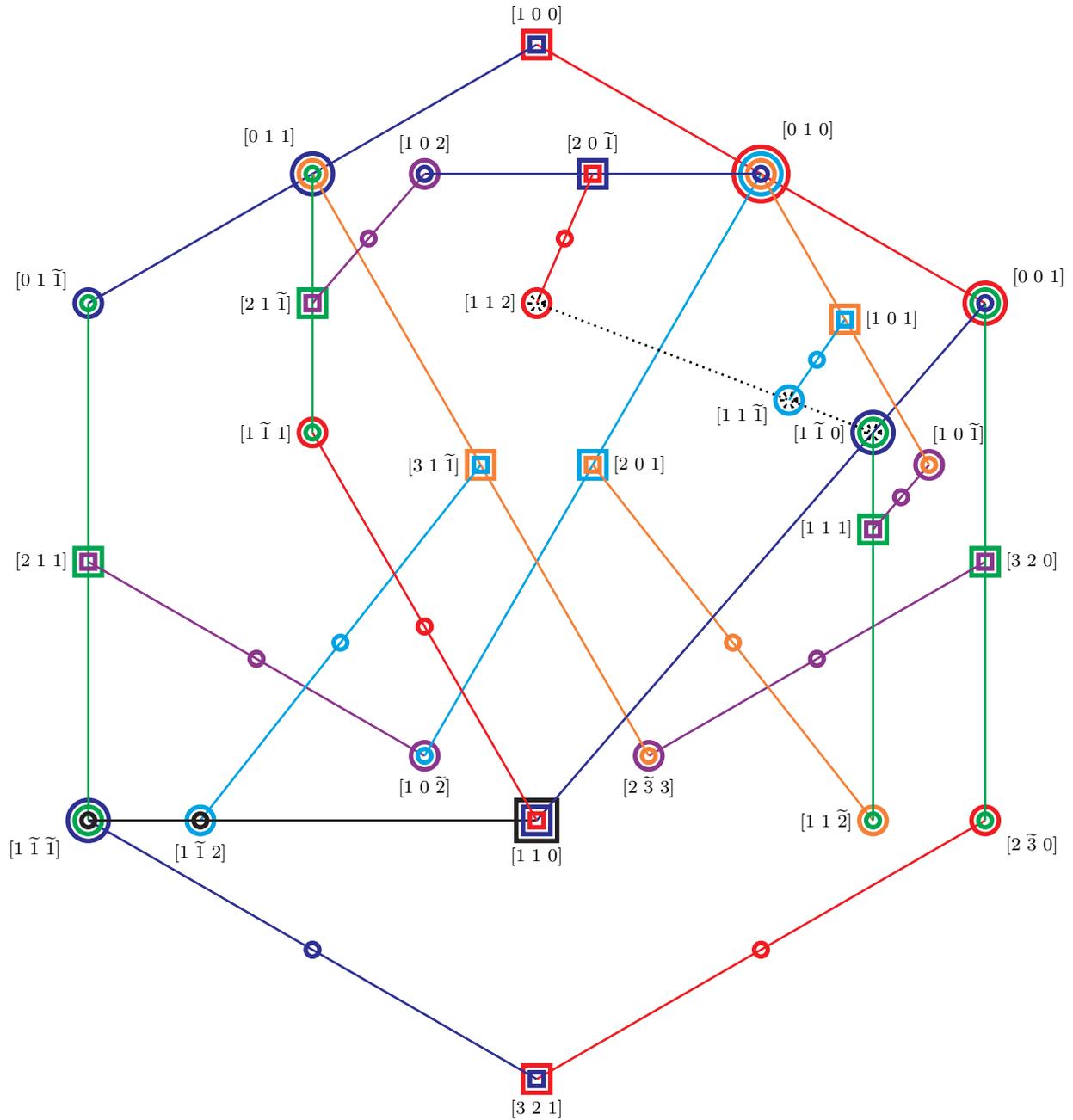
\begin{figure}
        \begin{center}
        \begin{tikzpicture}
        \tikzstyle{every path}=[line width=1pt]
        \tikzstyle{every node}=[draw,line width=2pt,inner sep=0]
        \tikzstyle{Hidden}=[draw opacity=0]

        \tikzstyle{y1}=[rectangle,minimum size=6]
        \tikzstyle{y2}=[rectangle,minimum size=12]
        \tikzstyle{y3}=[rectangle,minimum size=18]
        \tikzstyle{y4}=[rectangle,minimum size=24]

        \tikzstyle{n1}=[circle,minimum size=6]
        \tikzstyle{n2}=[circle,minimum size=12]
        \tikzstyle{n3}=[circle,minimum size=18]
        \tikzstyle{n4}=[circle,minimum size=24]

        \tikzstyle{l1}=[draw=none,rectangle,minimum size=8]
        \tikzstyle{l2}=[draw=none,rectangle,minimum size=14]
        \tikzstyle{l3}=[draw=none,rectangle,minimum size=20]
        \tikzstyle{l4}=[draw=none,rectangle,minimum size=26]

        \draw[Red] (90:8) -- (30:8)
                coordinate[y2,at start] (1 0 0)
                coordinate[n4,midway] (0 1 0)
                coordinate[n3,at end] (0 0 1);

        \draw[Blue] (1 0 0.center) -- (150:8)
                coordinate[y1,at start] (1 0 0)
                coordinate[n3,midway] (0 1 1)
                coordinate[n2,at end] (0 1 -1);

        \draw[Red] (270:8) -- (330:8)
                coordinate[y2,at start] (3 2 1)
                coordinate[n1,midway] (3 2 -13)
                coordinate[n2,at end] (2 -3 0);

        \draw[Blue] (3 2 1.center) -- (210:8)
                coordinate[y1,at start] (3 2 1)
                coordinate[n1,midway] (1 -4 5)
                coordinate[n3,at end] (1 -1 -1);

        \draw[Green] (1 -1 -1.center) -- (0 1 -1.center)
                coordinate[n2,at start] (1 -1 -1)
                coordinate[y2,midway] (2 1 1)
                coordinate[n1,at end] (0 1 -1);

        \draw[Green] (2 -3 0.center) -- (0 0 1.center)
                coordinate[n1,at start] (2 -3 0)
                coordinate[y2,midway] (3 2 0)
                coordinate[n2,at end] (0 0 1);

        \draw[Hidden] (2 1 1.center) -- (3 2 -13.center)
                coordinate[midway] (1 0 -2);

        \draw[Hidden] (3 2 0.center) -- (1 -4 5.center)
                coordinate[midway] (2 -3 3);

        \draw[Fuchsia] (2 1 1.center) -- (1 0 -2.center)
                coordinate[y1,at start] (2 1 1)
                coordinate[n1,midway] (2 -5 1)
                coordinate[n2,at end] (1 0 -2);

        \draw[Fuchsia] (3 2 0.center) -- (2 -3 3.center)
                coordinate[y1,at start] (3 2 0)
                coordinate[n1,midway] (6 -9 -13)
                coordinate[n2,at end] (2 -3 3);

        \draw[Cerulean] (0 1 0.center) -- (1 0 -2.center)
                coordinate[n3,at start] (0 1 0)
                coordinate[y2,midway] (2 0 1)
                coordinate[n1,at end] (1 0 -2);

        \draw[Orange] (0 1 1.center) -- (2 -3 3.center)
                coordinate[n2,at start] (0 1 1)
                coordinate[y2,midway] (3 1 -1)
                coordinate[n1,at end] (2 -3 3);

        \draw[Hidden] (1 -1 -1.center) -- (2 -3 0.center)
                coordinate[very near start] (1 -1 2)
                coordinate[midway] (1 1 0)
                coordinate[very near end] (1 1 -2);

        \draw[Orange] (2 0 1.center) -- (1 1 -2.center)
                coordinate[y1,at start] (2 0 1)
                coordinate[n1,midway] (1 -5 -2)
                coordinate[n2,at end] (1 1 -2);

        \draw[Cerulean] (3 1 -1.center) -- (1 -1 2.center)
                coordinate[y1,at start] (3 1 -1)
                coordinate[n1,midway] (1 -7 -4)
                coordinate[n2,at end] (1 -1 2);

        \draw[Black] (1 -1 -1.center) -- (1 1 0.center)
                coordinate[n1,at start] (1 -1 -1)
                coordinate[n1,near start] (1 -1 2)
                coordinate[y3,at end] (1 1 0);

        \draw[Blue] (1 1 0.center) -- (0 0 1.center)
                coordinate[y2,at start] (1 1 0)
                coordinate[n3,near end] (1 -1 0)
                coordinate[n1,at end] (0 0 1);

        \draw[Hidden] (0 1 -1.center) -- (2 -3 0.center)
                coordinate[near start] (1 -1 1);

        \draw[Red] (1 1 0.center) -- (1 -1 1.center)
                coordinate[y1,at start] (1 1 0)
                coordinate[n1,midway] (1 -1 -2)
                coordinate[n2,at end] (1 -1 1);

        \draw[Green] (1 -1 0.center) -- (1 1 -2.center)
                coordinate[n2,at start] (1 -1 0)
                coordinate[y2,near start] (1 1 1)
                coordinate[n1,at end] (1 1 -2);

        \draw[Green] (1 -1 1.center) -- (0 1 1.center)
                coordinate[n1,at start] (1 -1 1)
                coordinate[y2,midway] (2 1 -1)
                coordinate[n1,at end] (0 1 1);

        \draw[Hidden] (0 1 0.center) -- (3 2 0.center)
                coordinate[near end] (1 0 -1);

        \draw[Hidden] (0 1 0.center) -- (0 1 1.center)
                coordinate[near end] (1 0 2);

        \draw[Fuchsia] (1 1 1.center) -- (1 0 -1.center)
                coordinate[y1,at start] (1 1 1)
                coordinate[n1,midway] (1 -2 1)
                coordinate[n2,at end] (1 0 -1);

        \draw[Fuchsia] (2 1 -1.center) -- (1 0 2.center)
                coordinate[y1,at start] (2 1 -1)
                coordinate[n1,midway] (2 -5 -1)
                coordinate[n2,at end] (1 0 2);

        \draw[Orange] (1 0 -1.center) -- (0 1 0.center)
                coordinate[n1,at start] (1 0 -1)
                coordinate[y2,midway] (1 0 1)
                coordinate[n2,at end] (0 1 0);

        \draw[Blue] (1 0 2.center) -- (0 1 0.center)
                coordinate[n1,at start] (1 0 2)
                coordinate[y2,midway] (2 0 -1)
                coordinate[n1,at end] (0 1 0);

        \draw[Hidden] (1 0 0.center) -- (3 2 1.center)
                coordinate[near start] (1 1 2);

        \draw[Hidden] (1 1 2.center) -- (1 -1 0.center)
                coordinate[near end] (1 1 -1);

        \draw[Cerulean] (1 0 1.center) -- (1 1 -1.center)
                coordinate[y1,at start] (1 0 1)
                coordinate[n1,midway] (1 -2 -1)
                coordinate[n2,at end] (1 1 -1);

        \draw[Red] (2 0 -1.center) -- (1 1 2.center)
                coordinate[y1,at start] (2 0 -1)
                coordinate[n1,midway] (1 -5 2)
                coordinate[n2,at end] (1 1 2);

        \draw[dotted] (1 1 2.center) -- (1 -1 0.center)
                coordinate[n1,at start] (1 1 2)
                coordinate[n1,near end] (1 1 -1)
                coordinate[n1,at end] (1 -1 0);

        \coordinate[l2,label=90:\scriptsize\textcolor{black}{$[{1~0~0}]$}] (1 0 0) at (1 0 0.center);
        \coordinate[l4,label=60:\scriptsize\textcolor{black}{$[{0~1~0}]$}] (0 1 0) at (0 1 0.center);
        \coordinate[l3,label=30:\scriptsize\textcolor{black}{$[{0~0~1}]$}] (0 0 1) at (0 0 1.center);

        \coordinate[l3,label=120:\scriptsize\textcolor{black}{$[{0~1~1}]$}] (0 1 1) at (0 1 1.center);
        \coordinate[l2,label=150:\scriptsize\textcolor{black}{$[{0~1~\m1}]$}] (0 1 -1) at (0 1 -1.center);

        \coordinate[l2,label=270:\scriptsize\textcolor{black}{$[{3~2~1}]$}] (3 2 1) at (3 2 1.center);
        \coordinate[l2,label=330:\scriptsize\textcolor{black}{$[{2~\m3~0}]$}] (2 -3 0) at (2 -3 0.center);
        \coordinate[l3,label=210:\scriptsize\textcolor{black}{$[{1~\m1~\m1}]$}] (1 -1 -1) at (1 -1 -1.center);

        \coordinate[l2,label=180:\scriptsize\textcolor{black}{$[{2~1~1}]$}] (2 1 1) at (2 1 1.center);
        \coordinate[l2,label=0:\scriptsize\textcolor{black}{$[{3~2~0}]$}] (3 2 0) at (3 2 0.center);

        \coordinate[l2,label=270:\scriptsize\textcolor{black}{$[{1~0~\m2}]$}] (1 0 -2) at (1 0 -2.center);
        \coordinate[l2,label=270:\scriptsize\textcolor{black}{$[{2~\m3~3}]$}] (2 -3 3) at (2 -3 3.center);

        \coordinate[l2,label=0:\scriptsize\textcolor{black}{$[{2~0~1}]$}] (2 0 1) at (2 0 1.center);
        \coordinate[l2,label=180:\scriptsize\textcolor{black}{$[{3~1~\m1}]$}] (3 1 -1) at (3 1 -1.center);

        \coordinate[l2,label=180:\scriptsize\textcolor{black}{$[{1~1~\m2}]$}] (1 1 -2) at (1 1 -2.center);
        \coordinate[l2,label=270:\scriptsize\textcolor{black}{$[{1~\m1~2}]$}] (1 -1 2) at (1 -1 2.center);

        \coordinate[l3,label=270:\scriptsize\textcolor{black}{$[{1~1~0}]$}] (1 1 0) at (1 1 0.center);

        \coordinate[l3,label=180:\scriptsize\textcolor{black}{$[{1~\m1~0}]$}] (1 -1 0) at (1 -1 0.center);
        \coordinate[l2,label=180:\scriptsize\textcolor{black}{$[{1~\m1~1}]$}] (1 -1 1) at (1 -1 1.center);

        \coordinate[l2,label=180:\scriptsize\textcolor{black}{$[{1~1~1}]$}] (1 1 1) at (1 1 1.center);
        \coordinate[l2,label=180:\scriptsize\textcolor{black}{$[{2~1~\m1}]$}] (2 1 -1) at (2 1 -1.center);

        \coordinate[l2,label=87.5:\scriptsize\textcolor{black}{$[{1~0~\m1}]$}] (1 0 -1) at (1 0 -1.center);
        \coordinate[l2,label=90:\scriptsize\textcolor{black}{$[{1~0~2}]$}] (1 0 2) at (1 0 2.center);

        \coordinate[l2,label=0:\scriptsize\textcolor{black}{$[{1~0~1}]$}] (1 0 1) at (1 0 1.center);
        \coordinate[l2,label=90:\scriptsize\textcolor{black}{$[{2~0~\m1}]$}] (2 0 -1) at (2 0 -1.center);

        \coordinate[l2,label=185:\scriptsize\textcolor{black}{$[{1~1~\m1}]$}] (1 1 -1) at (1 1 -1.center);
        \coordinate[l2,label=180:\scriptsize\textcolor{black}{$[{1~1~2}]$}] (1 1 2) at (1 1 2.center);
        \end{tikzpicture}
        \end{center}
        \caption{(Color online) Greechie orthogonality diagram with an overlaid value assignment that can be used to visualise Table~\ref{2012-incomputability-proofs-table:greechie}.
                 The circles and squares represent observables that will be given the values $0$ and $1$ respectively.
                 They are joined by smooth lines which correspond to contexts, i.e.\ complete sets of compatible observables.}
        \label{2012-incomputability-proofs-fig:greechie}
        \end{figure}

        We now show that if $0<|\iprod{a}{b}|<\frac{3}{\sqrt{14}}\raisebox{.8mm}{,}$ and $P_a$ and $P_b$ both have the value 1, then there is a third observable $P_c$ which must also have the value 1 and satisfies $|\iprod{a}{c}|=\frac{3}{\sqrt{14}}\raisebox{.8mm}{.}$
        The above proof then applies to again show no admissible $v$ exists satisfying the requirements.

        By scaling $\ket{b}$ by a phase factor if necessary, we may assume that $\iprod{a}{b} \in \mathbb{R}$.
        Let $p = \iprod{a}{b}$ and $q = \sqrt{1 - p^2}$.
        Then $(\ket{b} - \ket{a} p) \frac{1}{q}$ is a unit vector orthogonal to $\ket{a}$.
        Taking a cross product, the set $\{ \ket{a}, (\ket{b} - \ket{a} p) \frac{1}{q}\raisebox{.8mm}{,} \ket{a} \times (\ket{b} - \ket{a} p) \frac{1}{q} \}$ forms an orthonormal basis for $\mathbb{C}^3$.
        Relative to this basis, $\ket{a} \equiv (1, 0, 0)$ and $\ket{b} \equiv (p, q, 0)$.
        Set $x = \frac{3}{\sqrt{14}}\raisebox{.8mm}{,}$ so that $p^2 < x^2$.
        Then $$\frac{p^2 (1 - x^2)}{q^2 x^2} = \frac{p^2 - p^2 x^2}{q^2 x^2} < \frac{x^2 - p^2 x^2}{q^2 x^2} = \frac{(1 - p^2) x^2}{q^2 x^2} = 1.$$
        Now set $y = \frac{p (1 - x^2)}{q x}\raisebox{.7mm}{,}$ so that $y^2 = \frac{p^2 (1 - x^2)}{q^2 x^2} (1 - x^2) < 1 - x^2$.
        Then we can set $z = \sqrt{1 - x^2 - y^2} \in \mathbb{R}$.
        This choice of $z$ makes $\ket{c} \equiv (x, y, z)$ a unit vector in $\mathbb{R}^3$.
        Taking cross products, we define
        \begin{align*}
                \ket{\alpha} &= \ket{a} \times \ket{c} \equiv (1, 0, 0) \times (x, y, z) = (0, -z, y), \\
                \ket{\beta} &= \ket{b} \times \ket{c} \equiv (p, q, 0) \times (x, y, z) = (q z, -p z, p y - q x),
        \end{align*}
        so that $\iprod{\alpha}{\beta} = (0, -z, y) \cdot (q z, -p z, p y - q x) = p z^2 + p y^2 - q x y = p (z^2 + y^2) - p (1 - x^2) = 0.$
        Therefore $\{ \ket{\alpha}, \ket{\beta}, \ket{c} \}$ is an orthogonal basis for $\mathbb{C}^3$.
        This implies that the projection observables $P_\alpha$, $P_\beta$, and $P_c$
        associated with the subspaces of $\mathbb{C}^3$ spanned by $\ket{\alpha}$, $\ket{\beta}$ and $\ket{c}$ are mutually compatible; that is,  $C_{25} = \{ P_\alpha, P_\beta, P_c \}$ is a context.
        Moreover, $P_\alpha$ is compatible with $P_a$ because $\iprod{\alpha}{a} = 0$.
        Likewise, $P_\beta$ is compatible with $P_b$.
        Hence there exist contexts $C_{26}$ and $C_{27}$ such that $P_\alpha, P_a \in C_{27}$ and $P_\beta, P_b \in C_{27}$.

        Define unit vectors $\ket{\psi} \equiv (0, 2 y - z, y + 2 z) \frac{\sqrt{14}}{5}$ and $\ket{\phi} \equiv (0, y + 2 z, z - 2 y) \frac{\sqrt{14}}{5}\raisebox{.4mm}{.}$
        Then it is easily checked that $\{ \ket{a}, \ket{\psi}, \ket{\phi} \}$ is an orthonormal basis for $\mathbb{C}^3$.
        Note that
        $$(\ket{a} 3 + \ket{\psi} 2 + \ket{\phi}) \tfrac{1}{\sqrt{14}} \equiv (\tfrac{3}{\sqrt{14}}\raisebox{.8mm}{,} (4 y - 2 z + y + 2 z) \tfrac{1}{5}\raisebox{.8mm}{,} (2 y + 4 z + z - 2 y) \tfrac{1}{5}) = (x, y, z) \equiv \ket{c},$$
         so $\ket{c} \equiv (3, 2, 1) \frac{1}{\sqrt{14}}$ relative to the basis $\{ \ket{a}, \ket{\psi}, \ket{\phi} \}$.

        Now let $\mathcal{C} = \{ C_1, C_2, \dots, C_{27} \}$ and $\mathcal{O} = \bigcup_{i = 1}^{27} C_i$.
        Suppose there exists an admissible assignment function $v$ under which $\mathcal{O}$ is non-contextual and $v(P_a, C_{26}) = v(P_b, C_{27}) = 1$.
        Since $v$ is admissible, it follows that $v(P_\alpha, C_{26}) = v(P_\beta, C_{27}) = 0$.
        Therefore $v(P_\alpha, C_{25}) = v(P_\beta, C_{25}) = 0$, so by admissibility $v(P_c, C_{25}) = 1$.
        This deduction is illustrated in Fig.~\ref{2012-incomputability-proofs-fig:greechie-mini}.
        However, by interpreting the observables in Table~\ref{2012-incomputability-proofs-table:greechie} as being defined relative to the basis $\{\ket{a},\ket{\psi},\ket{\phi}\}$, it is immediately clear that again no such admissible function $v$ exists.
\end{proof}

        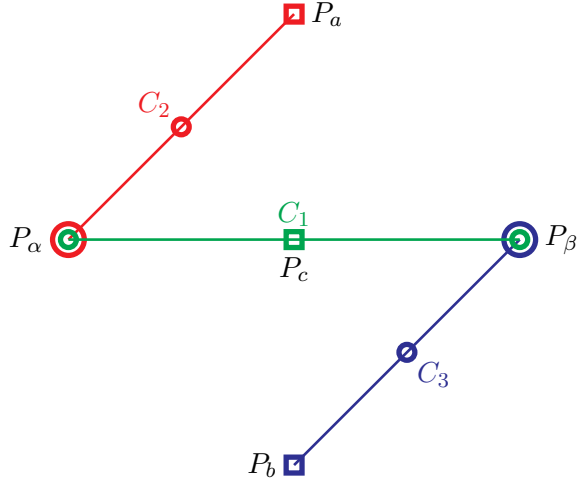
\begin{figure}[t]
        \begin{center}
        \begin{tikzpicture}
        \tikzstyle{every path}=[line width=1pt]
        \tikzstyle{every node}=[draw,line width=2pt,inner sep=0]

        \tikzstyle{y1}=[rectangle,minimum size=6]

        \tikzstyle{n1}=[circle,minimum size=6]
        \tikzstyle{n2}=[circle,minimum size=12]

        \tikzstyle{l1}=[draw=none,rectangle,minimum size=9]
        \tikzstyle{l2}=[draw=none,rectangle,minimum size=15]

        \draw[Red] (90:3) -- (180:3)
                coordinate[y1,at start] (a)
                coordinate[n1,midway,label=135:$C_2$] (0 y z)
                coordinate[n2,at end] (alpha);

        \draw[Blue] (270:3) -- (0:3)
                coordinate[y1,at start] (b)
                coordinate[n1,midway,label=315:$C_3$] (q -p ?)
                coordinate[n2,at end] (beta);

        \draw[Green] (alpha.center) -- (beta.center)
                coordinate[n1,at start] (alpha)
                coordinate[y1,midway,label=90:$C_1$] (c)
                coordinate[n1,at end] (beta);

        \coordinate[l1,label=0:$P_a$] (a) at (a.center);
        \coordinate[l1,label=180:$P_b$] (b) at (b.center);

        \coordinate[l2,label=180:$P_\alpha$] (alpha) at (alpha.center);
        \coordinate[l2,label=0:$P_\beta$] (beta) at (beta.center);

        \coordinate[l1,label=270:$P_c$] (c) at (c.center);
        \end{tikzpicture}
        \end{center}
        \caption{(Color online) Greechie orthogonality diagram with an overlaid value assignment that illustrates the relationship between the contexts $C_1,$ $C_2$ and $C_3$ in Theorem~\ref{thm:twonotvaluedefinite}.
                 The circles and squares represent observables that will be given the values $0$ and $1$ respectively.
                 They are joined by smooth lines which represent contexts.}
        \label{2012-incomputability-proofs-fig:greechie-mini}
        \end{figure}

\begin{Corollary}
\label{cor:twonotvaluedefinite}
        Let $\ket{a}, \ket{b}\in \mathbb{C}^3$ be unit vectors such that $\sqrt{\frac{5}{14}} \le |\iprod{a}{b}| \le \frac{3}{\sqrt{14}}\raisebox{.8mm}{.}$
        Then there exists a set of projection observables $\mathcal{O}$ containing $P_a$ and $P_b$, and a set of contexts $\mathcal{C}$ over $\mathcal{O}$, such that there is no admissible assignment function under which $\mathcal{O}$ is non-contextual, $P_a$ has the value 1 and $P_b$ is value definite.
\end{Corollary}
\begin{proof}
        Again scale $\ket{b}$ so that $\iprod{a}{b} \in \mathbb{R}$.
        Let $p = \iprod{a}{b}$ and $q = \sqrt{1 - p^2}$.
        As above we construct an orthonormal basis in which $\ket{a} \equiv (1,0,0)$ and $\ket{b} \equiv (p,q,0)$.
        Define $\ket{\alpha} \equiv (0,1,0)$, $\ket{\beta} \equiv (0,0,1)$ and $\ket{c} \equiv (q,-p,0)$.
        Then $\{ \ket{a}, \ket{\alpha}, \ket{\beta} \}$ and $\{ \ket{b}, \ket{c}, \ket{\beta} \}$ are orthonormal bases for $\mathbb{C}^3$, so we can define the contexts $C_1 = \{ P_a, P_\alpha, P_\beta \}$ and $C_2 = \{ P_b, P_c, P_\beta \}$.
        Note that $p^2 \ge \frac{5}{14}$ and hence $$\iprod{a}{c} = q = \sqrt{1 - p^2} \le \sqrt{1 - \tfrac{5}{14}} = \tfrac{3}{\sqrt{14}}\raisebox{.8mm}{.}$$
        From Theorem~\ref{thm:twonotvaluedefinite} it follows that there are sets of observables $\mathcal{O}_b$, $\mathcal{O}_c$ and contexts $\mathcal{C}_b$, $\mathcal{C}_c$ such that there is no admissible assignment function under which $\mathcal{O}_b$ ($\mathcal{O}_c$) is non-contextual and $P_a, P_b$ ($P_a, P_c$) have the value $1$.
        We combine these sets to give $\mathcal{O} = \mathcal{O}_b \cup \mathcal{O}_c \cup \{ P_\alpha, P_\beta \}$ and $\mathcal{C} = \mathcal{C}_b \cup \mathcal{C}_c \cup \{ C_1, C_2 \}$.
        Suppose there exists an admissible assignment function $v$ under which $\mathcal{O}$ is non-contextual, $v(P_a, C_1) = 1$ and $P_b$ is value definite.
        Then $v(P_b, C_2) \neq 1$ by the definition of $\mathcal{O}_b$, so $v(P_b, C_2) = 0$.
        Since $v(P_a, C_1) = 1$ and $v$ is admissible, $v(P_\beta, C_1) = 0$ and hence $v(P_\beta, C_2) = 0$ as well.
        So by admissibility $v(P_c, C_2) = 1$, which is impossible by the definition of $\mathcal{O}_c$.
        Therefore there does not exist such a function $v$.
\end{proof}

The difference between the above result and the Kochen-Specker theorem is subtle but critical.
The Kochen-Specker theorem, under the assumption of non-contextuality, only finds a contradiction with the hypothesis that \emph{all} observables are value definite---it does not allow any specific observable to be proven value indefinite.
Corollary~\ref{cor:twonotvaluedefinite}, however, allows just this---{\em specific value indefinite observables can be identified}.
While we delay the physical interpretation of this result until the following section, we mention that it applies to measurements of an observable on a physical system in an eigenstate of a different observable.

\section{Physical Interpretation}
\label{sec:interpretation}


In order to make operational use  of the results of the previous section
we connect the formal entities with measurement outcomes.
In the process of doing this, we make explicit the assumptions that our results rely on.

\subsection{The role of measurement}

An inherent assumption in the attempt to attribute physical meaning to the Kochen-Specker theorem (as well as the other theorems we have proved), and one which we shall also make, is that measurement is actually a physically meaningful process.
In particular, we assume:\\

        {\em Measurement assumption.} Measurement yields a physically meaningful and unique result.

This may seem rather self-evident, but it is not true of interpretations of quantum mechanics such as the many-worlds
interpretation, where measurement is just a process by which the apparatus or experimenter becomes entangled with the state being ``measured.''
In such an interpretation it does not make sense to talk about the unique ``result'' of a measurement, let alone any definite values which one may pre-associate with them.
\if01
{\color{blue}
Another scenario in which the measurement result is`meaningless' with regards to the physical object observed is
the \emph{context translation principle} \cite{svozil-2003-garda}, according to which,
in the case of a mismatch between preparation and measurement contexts,
the outcome is strongly influenced by the many degrees of
freedom of a macroscopic measurement apparatus, and therefore, as suggested by Bell \cite[p.~451]{bell-66}
``the result of an observation may reasonably depend $\ldots$ on the complete disposition of the apparatus''.
}

To establish the relationship between the quantum system of interest and the function $v$ assigning definite values in advance, we need to restrict ourselves to assignment functions which agree with quantum mechanics.
Specifically, definite values prescribed by the function should be just that; they must guarantee the result of a measurement. We thus have:\\

\begin{Assumption}
If, for a given physical system, an observable $o$ is value definite in a context $C\in\mathcal{C}$ under $v$, then the result of a measurement of $o$ in the context $C$ yields the value $v(o,C)$ with certainty.
Furthermore, the state of the system ``collapses'' into the corresponding eigenstate of $o$.
\end{Assumption}

Any assignment function not meeting this condition surely does not capture the notion of assigning definite values and is hence of little interest.
Of course, an assignment function which is defined nowhere meets this condition, but this complete indefiniteness does not fully capture our knowledge of a quantum system; we should at least be able to predict the outcomes of \emph{some} measurements.
We discuss this issue of when to assign definite values in Sec~\ref{sec:predImpliesVI}.
\fi

To establish the relationship between the quantum system of interest and the function $v$ assigning definite values in advance, we need to restrict ourselves to assignment functions which agree with quantum mechanics.
Specifically, definite values prescribed by the function should be just that; they must guarantee the result of a measurement.\\
        Let $v$ be a value assignment function. We say that $v$ is a \emph{faithful} representation of a realization $r_\psi$ of a state $\ket{\psi}$ if a measurement of observable $o$ in the context $C$ on the physical state $r_\psi$ yields the result $v(o,C)$ whenever $o$ has a definite value under $v$.
Usually, it is implicitly assumed that a value assignment function is faithful---if it is not then it has no real relation to the physical system that it is meant to model and is of little interest.
Nonetheless, since we intend to make all assumptions explicit here, we will make clear that we are referring to faithful assignment functions when necessary.
Of course, an assignment function which is defined nowhere meets this condition, but this complete indefiniteness does not fully capture our knowledge of a quantum system; we should at least be able to predict the outcomes of \emph{some} measurements.
We discuss this issue of when to assign definite values in Sec~\ref{sec:predImpliesVI}.


\subsection{Value indefiniteness}

The Kochen-Specker theorem leaves two possibilities: either we give up the idea that every observable should be simultaneously value definite, or we allow observables to be defined contextually. Of course, some combination of both options is also possible.
Here we opt to assume non-contextuality of observables for which the outcome is predetermined, and thus give up the historic notion of complete determinism (classical omniscience).

This assumption might be in contradiction to that of  physicists who,
in the tradition of the realist Bell (see the oft-quoted text,~\cite{bell-66}),
tend to opt for contextuality.
The option for contextuality saves realistic omniscience and ``contextual value definiteness'' at the price of introducing a more general dependence of at least some potential observables on the measurement context.
In what follows we make no attempt to save realistic omniscience  and instead require the non-contextuality of any pre-determined properties.

        {\em Non-contextuality assumption.} The set of observables $\mathcal{O}$ is non-contextual.

While from the Kochen-Specker theorem and our discussion of strong-contexuality it is mathematically conceivable that only some observables are forced to be value indefinite, while others remain both non-contextual and value definite,
this would be a rather strange scenario due to the overall uniformity and symmetry of these arguments.
Regardless, if we can guarantee that one observable $P_a$ is value definite, with the value 1 (e.g.\ by preparing the system in an eigenstate of $P_a$ with eigenvalue 1), Corollary~\ref{cor:twonotvaluedefinite} gives us some observables that must be value indefinite.

\subsection{Predictability implies value definiteness}
\label{sec:predImpliesVI}

A more subtle assumption relates to the question of when we should consider a physical observable to have a definite value associated with it, and the connection between these definite values and probability.
Einstein, Podolsky and Rosen (EPR), in their seminal paper on the EPR paradox as it is now known, said~\cite[p. 777]{epr}:
\begin{quote}
        If, without in any way disturbing a system, we can predict with certainty (i.e., with probability equal to unity)
the value of a physical quantity, then there exists an element of physical reality \footnote{An element of physical reality corresponds
to the notion of a definite value, possibly contextual,
as outlined in this paper.} [(e.p.r.)]  corresponding to this physical quantity.
\end{quote}

From the physicist's point of view, the ability to predict the value of an observable with certainty seems sufficient to posit the existence of a definite value associated with that observable. However, the identification that EPR make between certainty and probability one is less sound.
Mathematically, the statement is simply not true: for infinite  measure spaces probability zero events not only can, but must occur---every point has probability 0 under the Lebesgue measure. With a frequentist view of probability, the two notions cannot be united even for finite spaces. One can only say an event is certain if its complement is the empty set.

With the formalism of quantum mechanics entirely based on probability spaces, what then can we say about any definite values in physical reality?
A deterministic theory is based on a description of a state which is complete in that it specifies definite values for all observables.
The state in quantum theory, however, is given as a wave function, which in turn is determined by the operators
of which the system is an eigenstate. Quantum theory is thus based on the notion that a physical state is
``completely characterised by the wave function,''
which is an eigenstate of some operator and is determined for any context containing the said operator; as EPR note,
the ``physical quantity'' corresponding to that operator has ``with certainty'' the corresponding eigenvalue~\cite[p. 778]{epr}.
The theory then presents a probabilistic framework to express behavior in other contexts.
A reasonable assumption based on this principle is the following:

{\em Eigenstate assumption.} Let $\ket{\psi}$ be a (normalised) quantum state and $v$ be a faithful assignment function. Then $v(P_\psi,C)=1$ and $v(P_\phi,C)=0$ for
any context $C\in\mathcal{C}$ with $P_\psi,P_\phi\in C$.

While this is a reasonable condition under which to assign an initial set of definite values,
its use is restricted to contexts containing the ``preparation'' observable.
In order to extend this, we must more carefully formulate the notion of being able to predict the value of an observable with certainty.


Let us consider a system which we prepare, measure, rinse and repeat ad infinitum.
Let $\seq{x}=x_1 x_2 \dots$ denote the infinite sequence produced by concatenating the outputs of these measurements.
Fix a set of observables $\mathcal{O}$ and contexts $\mathcal{C}$ and let $o_i,C_i$ denote the observable and corresponding context of the $i$th measurement.
We can predict with certainty the value of each measurement if there exists a computable function
$f : \mathbf{N}\times \mathcal{O}\times \mathcal{C} \to \{0,1\}$ such that, for every $i$, $f(i,o_i,C_i)=x_i$.
Why do we require that $f$ be computable?
Since we must with every measurement obtain a result, there is guaranteed to be some function giving $\seq{x}$ from the measurements, but if it is not computable then this function offers no method to predict the values.
Why do we formulate this for infinite sequences?
The notion of computability, and thus concrete predictability, only makes sense for infinite sequences; it is clear that any technique which allows prediction of every measurement with certainty must also do so when the measurements are continued ad infinitum.

The last assumption is the

{\em Elements of physical reality (e.p.r.) assumption.} If
there exists a computable function $f : \mathbf{N}\times \mathcal{O}\times \mathcal{C} \to \{0,1\}$
such that for every $i$ $f(i,o_i,C_i)=x_i$, then there is a definite value associated with $o_i$ at each step
[i.e., $v_i(o_i,C_i)=f(i,o_i,C_i)$].

We note that the assumption above does not postulate the existence of an effective way to find or to compute
 the computable function $f$: such a function simply exists.
This is visible in classical hidden variable type theories such as   statistical mechanics for thermodynamics, where we can hardly claim to be able to even describe fully the momentum and position of each particle in a gas, but it is sufficient to know that we \emph{can} do so and that these hidden variables exist in the sense that they allow us, in principle,  to predict the outcome of any measurement in advance.
Furthermore, we follow EPR in noting that this is certainly only a sufficient condition for definite values to be present; it is by no means necessary.


\subsection{Connection to quantum theory}

The final step is to justify our requirement of the admissibility of the assignment function.


We first note the following:
Let $C=\{P_1,\dots,P_n\}$ be a context of projection observables, $v$ a faithful assignment function and $v(P_1,C)=1$.
\if01
Since $\oprod{\psi}{\psi}$ and $\oprod{\phi}{\phi}$ are co-measurable, if we measure them both, the system will collapse into an eigenstate of $\oprod{\psi}{\psi}$,
corresponding to the eigenvalue 1, which is orthogonal to $\ket{\phi}$.
Since this final state would also be an eigenvector of $\oprod{\phi}{\phi}$, Lemma~\ref{lem:spectrum} implies that it corresponds to 0.
Therefore this measurement of $\oprod{\phi}{\phi}$ is predetermined to report the value 0, so by non-contextuality any measurement of $\oprod{\phi}{\phi}$ will have the same result.
\fi
Since $P_1$ and $P_i$ ($i\neq 1$) are compatible (physically co-measurable), if we measure them both, the system will collapse into the eigenstate of $P_1$
 corresponding to the eigenvalue 1.
Since this final state would also be an eigenstate of $P_i$,
 it follows from the fact that $\sum_{j=1}^n P_j=\mathbf{1}$ that this state corresponds to the eigenvalue 0 of $P_i$ and hence $v(P_i,C)=0$.
Hence we conclude that $v(P_i,C)=0$ for all $2\le i \le n$.
By a similar argument, we see that if instead $v(P_i,C)=0$ for $2\le i \le n$ we must have $v(P_1,C)=1$.
\if01
From Lemma~\ref{lem:compatible}, we know that $\mathcal{P}$ is a collection of mutually co-measurable observables.
If we measure them all, the system will collapse into an eigenstate (with eigenvalue 0) of every operator in $\mathcal{P}$, which by Lemma~\ref{lem:spectrum} will be orthogonal to every state in $\mathcal{B} - \{\ket{\psi}\}$.
Let $\ket{\phi}$ be this state.
Then $\ket{\phi} = \sum_{\ket{\beta} \in \mathcal{B}} \ket{\beta} \iprod{\beta}{\phi} = \ket{\psi} \iprod{\psi}{\phi}$, which implies (by Lemma~\ref{lem:compatible} again) that the eigenvalue of $\ket{\phi}$ for $\oprod{\psi}{\psi}$ is 1.
Therefore this measurement of $\oprod{\psi}{\psi}$ (and any other one, by non-contextuality) will report the value 1.
\fi

        From these facts it follows directly that a faithful assignment function $v$ must be admissible,
thus justifying our definition of an admissible $v$.
Indeed, admissibility of $v$ is the direct generalisation of the ``sum rule'' used in proofs of the Kochen-Specker theorem~\cite{kochen1,Peres:1996fk} to the case where value definiteness is not assumed.
In our deduction of the requirement of admissibility we are particularly careful in using our assumptions to show that admissibility is required if simple relations of projection observables are to be satisfied.\\

As a consequence, we get the following useful form of Corollary~\ref{cor:twonotvaluedefinite} which we will utilize in the remainder of the paper.\\
\begin{Corollary}
        \label{col:VIContext}
        Let $\ket{\psi} \in \mathbb{C}^3$ be a quantum state describing a system.
        Also let $\ket{\phi} \in \mathbb{C}^3$ be any other state which satisfies $\sqrt{\frac{5}{14}} \le |\iprod{\psi}{\phi}| \le \frac{3}{\sqrt{14}}\raisebox{.8mm}{.}$
        Then, assuming non-contextuality, $P_\phi$ cannot be assigned a definite value by a faithful assignment function.
\end{Corollary}
\begin{proof}
        From the Eigenstate assumption, $P_\psi$ must be assigned the value 1. By Corollary~\ref{cor:twonotvaluedefinite} and the requirement for a faithful assignment function to be admissible, it follows that $P_\phi$ must be value indefinite.
\end{proof}

\section{A Random Number Generator}

From our assumptions of non-contextuality along with our physical assumptions in the preceding section, we arrived at the key result of Corollary~\ref{col:VIContext}, which allows us to identify particular observables which must be value indefinite.
This guarantee of indefiniteness, which both the Bell~\cite{bell-66} and Kochen-Specker theorems cannot yield, adds extra conviction to the widely accepted (but not proven) unpredictability of the result of quantum measurements.
Since quantum random number generators (QRNGs)~\cite{Kwon:09,10.1038/nature09008,stefanov-2000,svozil-2009-howto,Wayne-09,stipcevic045104,Ma:05} depend entirely on this, it seems clear we should make use of this extra certification in their design.
In this section we present such a design of a QRNG, and use Corollary~\ref{col:VIContext} to prove that such a device will produce strongly incomputable sequences of bits---a strong, explicit certification of the QRNG.

\subsection{Random number generator design}

The QRNG setup is shown in Fig.~\ref{fig:setup}.
Spin-1 particles are prepared in the $\textsf{S}_z=0$ state (thus, by the Eigenstate assumption, this operator has a definite value),
and then the $\textsf{S}_x$ operator is measured.
Since the preparation state is an eigenstate of the $\textsf{S}_x=0$ projector with eigenvalue 0,
this outcome has a definite value and cannot be obtained.
Thus, while the setup uses spin-1 particles, the outcomes are dichotomic and the $\textsf{S}_x=\pm 1$ outcomes can be assigned 0 and 1 respectively.
Furthermore, since $\iprod{\textsf{S}_z=0}{\textsf{S}_x=\pm 1}=1/\sqrt{2}$, it follows from Corollary~\ref{col:VIContext} that neither of the $\textsf{S}_x=\pm 1$ outcomes can have pre-assigned definite value.

While this design is very simple, it has the two key properties we need from such a QRNG: it produces bits certified by value indefiniteness, and it produces the bits 0 and 1 independently and with 50/50 probability.

\begin{figure}

\begin{center}
\begin{tikzpicture}
\tikzstyle{c}=[circle,draw,fill=white]
\tikzstyle{r}=[rectangle,draw,fill=white,inner sep=7pt]

\node[r] (source) at (0,0) {spin-1 source};
\node[r] (z) at (4,0) {$\textsf{S}_z$ splitter};
\node[r] (x) at (8,0) {$\textsf{S}_x$ splitter};

\draw[->] (source) -- (z);

\draw[->] (z) edge node[above] {\scriptsize 1} (6.8,1);
\draw[->] (z) edge node[above] {\scriptsize 0} (x);
\draw[->] (z) edge node[above] {\scriptsize -1} (6.8,-1);

\node[c] (1) at (11.5,1) {\texttt{1}};
\node[c] (0) at (11.5,-1) {\texttt{0}};

\draw[->] (x) edge node[above] {\scriptsize 1} (1);
\draw[->,dotted] (x) edge node[above] {\scriptsize 0} (11.2,0);
\draw[->] (x) edge node[above] {\scriptsize -1} (0);
\end{tikzpicture}
\end{center}
\caption{Experimental setup of a configuration of quantum observables rendering random bits certified by quantum value indefiniteness.}
        \label{fig:setup}
\end{figure}
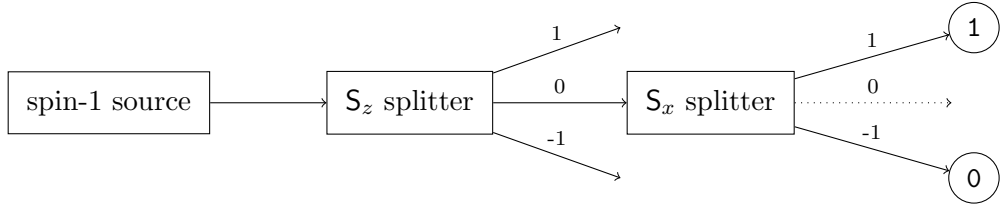

\subsection{Certification via value indefiniteness}

Consider the QRNG
described in the previous section, and let us consider that we run it repeatedly ``to infinity;''  that is, we use it repeatedly to generate bits and concatenate them together to produce, in the limit, the binary sequence $\seq{x}=x_1 x_2 \dots x_n \dots$.
Here we consider the sequence $\seq{x}$ produced in such a manner and show that, under our assumptions, it is guaranteed to be {\em incomputable}.
Note that we are using the measurement assumption here, since we must assume that $\seq{x}$ is actually produced (not that, for example, all infinite sequences are generated in different universes).

Before presenting our argument we note that
Martin-L\"of's theorem in algorithmic information theory~\cite{calude:02}
shows that {\em there are no  pure, true or perfect random sequences}:
there are patterns in every sequence,
a deterministic provable fact which is much stronger than the typical
highly probable results (facts true with probability one) proved in probability theory.
Because  we cannot speak about pure, true or perfect randomness
we have no option but to study degrees and symptoms of randomness:
some sequences are more random than others.
Uniform distribution within a sequence (Borel normality \cite{borel:09}) is a symptom of randomness:
however,  there exist
computable uniformly distributed sequences, (e.g.,  the Champernowne
 sequence~\cite{calude:02}),  which are far from being random in any meaningful way.
Unpredictability is another symptom;
(strong) incomputability is one mathematical way to express it.
Uniform distribution and unpredictability
are independent;
while the lack of uniform distribution can be easily mitigated by procedures
\`a la von Neumann~\cite{AbbottCalude10},
transforming a computable sequence  into an incomputable one is a much more difficult problem.

Quantum randomness is usually qualified in terms of the probability distribution of the source.
This only allows for probabilistic claims about the outcomes of individual measurements.
For example, with probability one any sequence of quantum random bits is incomputable; such a statement is
 weaker than saying that the sequence is  provably incomputable.
Nevertheless, claims made in different articles,
even recent ones such as Refs.~\cite{10.1038/nature09008,Merali:2014aa}
or web sites \cite{idQuantique,ANU}, according to which
``perfect randomness can be obtained via quantum experiments,''
are only of this statistical nature.
Here we are able to prove the guaranteed  incomputability of quantum randomness;
however, due to
Martin-L\"of's theorem, even this result cannot be called ``perfect randomness.''

For the sake of contradiction let us assume that $\seq{x}$ as described above is computable.
Then, by definition, there must exist a Turing machine $T$ (and thus a computable function) that can be associated with $\seq{x}$ allowing us to predict with certainty every value $x_i$.
From the e.p.r. assumption, it follows that each observable $o_i$ is value definite and $v_i(o_i,C)=x_i$.
This contradicts the implications of Corollary~\ref{col:VIContext}.
Thus we conclude that $\seq{x}$ must be incomputable.

This proof can easily show the stronger claim: that $\seq{x}$ is {\em bi-immune};  that is, no infinite sub-sequence of $\seq{x}$ is computable.
This can easily by seen by the same argument: if there was a computable subsequence then we could assign definite values to the observables giving rise to this subsequence, contradicting our assumption of value indefiniteness everywhere.


We have proved:\\

Assume the non-contextuality, measurement, eigenstate and e.p.r. assumptions.  Then there exits a QRNG which generates a bi-immune binary sequence.

We further note that this result is more general than that proved in Ref.~\cite{svozil-2006-ran}
and does not require any assumption about the uniformity of the bits produced.

\subsection{Experimental robustness}
\label{2012-incomput-proofsCJ-er}

Before we proceed to describe an explicit realization of the QRNG described above, we wish to briefly make a couple of points on the robustness of this certification by value indefiniteness to experimental imperfections.

We can describe the measurement context more generally by the spin observable $\textsf{S}(\theta,\phi)$,
where $\theta$ and $\phi$ are the polar and azimuthal angles respectively, and we thus have $\textsf{S}_x=\textsf{S}(\pi/2,0)$
and
 $\textsf{S}_z=\textsf{S}(0,0)$.
Explicitly, this operator is represented in matrix form as
\begin{equation}
\textsf{S}(\theta,\phi )
=
\left(
\begin{array}{ccc}
 \cos (\theta ) & \frac{e^{-i \phi } \sin (\theta )}{\sqrt{2}} & 0 \\
 \frac{e^{i \phi } \sin (\theta )}{\sqrt{2}} & 0 & \frac{e^{-i \phi } \sin (\theta )}{\sqrt{2}} \\
 0 & \frac{e^{i \phi } \sin (\theta )}{\sqrt{2}} & -\cos (\theta )
\end{array}
\right).
\end{equation}
Misalignment and imperfection in the experimental setup will, in general, lead to angles $\theta$ and $\phi$ differing slightly from $\pi/2$ and 0 respectively.
While a change in $\phi$ only induces a phase-shift and does not alter the probability of measuring any particular eigenvalue, a change in $\theta$ will alter the probabilities of detection.
However, a detailed calculation shows that \begin{equation}\label{eqn:genProb}|\iprod{\textsf{S}_z=0}{\textsf{S}(\theta,\phi)=\pm 1}|=\sin\theta/\sqrt{2},\end{equation} and the difference in probabilities of measuring a bit as 0 or 1 is not affected by such a change in $\theta$.
This is in distinct contrast to setups based on single beam splitters, in which misalignment introduces bias into the distribution of bits.

From Corollary~\ref{col:VIContext}, we see that the QRNG will provide bits by measurement of $\textsf{S}(\theta,\phi)$ that are certified by value indefiniteness whenever $\sqrt{\frac{5}{14}} \le |\iprod{\textsf{S}_z=0}{\textsf{S}(\theta,\phi)=\pm 1}| \le \frac{3}{\sqrt{14}}\raisebox{.8mm}{.}$
This inequality is, from Eq.~(\ref{eqn:genProb}), readily seen to be satisfied for angles $\frac{\pi}{3} \le \theta \le \frac{2 \pi}{3}\raisebox{.8mm}{.}$
This has the important consequence of protecting against inevitable experimental misalignment: even in the presence of relatively significant misalignment, the device would produce bits which are certified by value indefiniteness.
Otherwise, if the certification only held for the ideal case of $\frac{\pi}{2}\raisebox{.8mm}{,}$ any experimental imperfections would render this theoretical result inapplicable to any real experiment.

Furthermore, calculation shows that $\iprod{\textsf{S}_z=0}{\textsf{S}(\theta,\phi)=0}=\cos\theta$, and since $\iprod{\textsf{S}_z=0}{\textsf{S}(\theta,\phi)=0}=0$ only when $\theta=\frac{\pi}{2}\raisebox{.8mm}{,}$ a third detector measuring the $\ket{\textsf{S}(\theta,\phi)=0}$ outcome could be employed to monitor the degree of misalignment present in the system.
The number of counts at this detector would allow quantification of the angle $\theta$,
and provide an experimental method to test that the condition of
$\sqrt{\frac{5}{14}} \le \iprod{\textsf{S}_z=0}{\textsf{S}(\theta,\phi)=\pm 1}\le \frac{3}{\sqrt{14}}$ is indeed being realized.
Without monitoring this third outcome, one could not determine from the $\ket{\textsf{S}(\theta,\phi)=\pm 1}$ counts alone if this is indeed the case.

\section{Generalised Beam Splitter Quantum Random Number Generator}
\label{sec:experiment}

In this section we describe a physical realization of the QRNG described in the previous section.
Since it is not particularly feasible to directly use spin-1 particles
in a QRNG with an acceptably high bit-rate, the realization we present uses photons and is expressed
in terms of generalised beam splitters \cite{rzbb,zukowski-97,svozil-2004-analog}.
Generalised beam splitters are based on the possibility to (de)compose
an arbitrary unitary transformation $\textsf{U}_n$ in $n$-dimensional Hilbert space
into two-dimensional transformations $\textsf{U}_2$ of two-dimensional subspaces thereof;
a possibility that can be used to parametrize $\textsf{U}_n$ \cite{murnaghan}.
In more physical terms, they amount to serial stacks of phase shifters and beam splitters
in the form of an interferometer with $n$ input and output ports,  beam splitter
such that the beam  splitters affect only two (sub\nobreakdash-)paths
which, together with the phase shifters (affecting single paths at any one time),
realize the associated transformations in $U(2)$.
These components can be conveniently arranged into ``triangle form''
with $n$ in- and out-bound beam paths.

For the sake of an explicit demonstration, consider an orthonormal cartesian standard basis
$\vert 1\rangle \equiv (1,0,0)$,
$\vert 0\rangle \equiv  (0,1,0)$, and
$\vert -1\rangle \equiv  (0,0,1)$.
Then, in order to realize observables such as the spin state observables
$S(\theta,\phi)$
and, in particular, spin states  measured along the $x$-axis;
that is, for $\theta = \frac{\pi}{2}$ and $\phi = 0$,
\begin{equation}
\textsf{S}_x   =
\textsf{S}\left(\frac{\pi}{2},0 \right)
=
\left(
\begin{array}{ccc}
 0 & \frac{1}{\sqrt{2}} & 0 \\
 \frac{1}{\sqrt{2}} & 0 & \frac{1}{\sqrt{2}} \\
 0 & \frac{1}{\sqrt{2}} & 0
\end{array}
\right)
\end{equation}
in terms of generalised beam  splitters, the
associated normalised row eigenvectors
\begin{equation}
\begin{array}{l}
 \ket{\textsf{S}_x : +1} \equiv  \frac{1}{2}\left(1,\sqrt{2},1\right),\\
 \ket{\textsf{S}_x : 0} \equiv  \frac{1}{\sqrt{2}}\left(1,0,-1\right),\\
 \ket{\textsf{S}_x : -1} \equiv  \frac{1}{2}\left(1,-\sqrt{2},1\right)
\end{array}
\end{equation}
have to be ``stacked'' on top of one another \cite{rzbb}, thereby forming a unitary matrix $\textsf{U}_x $
which corresponds to the spin state operator $\textsf{S}_x$ for spin state measurements along the $x$-axis;
more explicitly,
\begin{equation}
\textsf{U}_x
=
\frac{1}{2}
\left(
\begin{array}{cccc}
1&\sqrt{2}&1\\
\sqrt{2}&0&-\sqrt{2}\\
1&-\sqrt{2}&1
\end{array}
\right).
\end{equation}

\if01
A beam splitter and phase shifter pair can be represented by the unitary matrix
\begin{equation}
        \begin{pmatrix}
e^{i\phi}\sin\omega & e^{i \phi}\cos\omega\\
\cos\omega & -\sin\omega
\end{pmatrix},
\end{equation}
where $\phi$ is represents the phase of the external phase shifter on one of the output ports,
and $\omega$ relates to the transmittance of the beam splitter by the equation $T=\sin^2\omega$.
The beam splitter arrangement to realize $\textsf{U}_x$
can be found by transforming $U_x$ into the identity matrix $I_3$ by successive right multiplication by $U(2)$ matrices---each one making an individual off-diagonal element equal to zero~\cite{rzbb}.

In our specific case, we have
\begin{equation}
        \textsf{U}_x \cdot
        \begin{pmatrix}
        1 & 0 & 0\\
        0 & \frac{1}{\sqrt{3}} & \sqrt{\frac{2}{3}}\\
        0 & \sqrt{\frac{2}{3}} & -\frac{1}{\sqrt{3}}
        \end{pmatrix}
        \cdot
        \begin{pmatrix}
        \frac{\sqrt{3}}{2} & 0 & \frac{1}{2}\\
        0 & 1 & 0\\
        \frac{1}{2} & 0 & -\frac{\sqrt{3}}{2}
        \end{pmatrix}
        \cdot
        \begin{pmatrix}
        \frac{1}{\sqrt{3}} & \sqrt{\frac{2}{3}} & 0\\
        \sqrt{\frac{2}{3}} & -\frac{1}{\sqrt{3}} & 0\\
        0 & 0 & 1
        \end{pmatrix}
        =I_3.
\end{equation}
This corresponds to three beam splitters with $$\omega_{3,2}=\omega_{2,1}=\tan^{-1}\left(\frac{1}{\sqrt{2}}\right),\ \omega_{3,1}=\frac{\pi}{3}\raisebox{.8mm}{,}$$
where $\omega_{i,j}$ is the parameter for the beam  splitter operating on beams $i$ and $j$ (beams 1,2,3 correspond to $\textsf{S}_z=+1,0,-1$ respectively). No external phase shifters are required as each $\phi=0$.
The transmittance parameters give transmittances of $T_{3,2}=T_{2,1}=1/3$ and $T_{3,1}=3/4\raisebox{.8mm}{.}$
The physical realization of $\textsf{U}_x$ is depicted in Fig.~\ref{2012-incomput-proofsC-f1}.

\begin{figure}[ht]
\begin{center}

\unitlength 20mm
\linethickness{0.8pt}
\begin{picture}(6,5)(-3.0,-1.5)

\put(-2.6,3){\line(0,-1){3}}

\put(-2.6,2.5){\line(1,1){0.2}}
\put(-2.6,2.5){\line(-1,1){0.2}}

\put(-2.6,2){\line(1,0){1}}
\put(-2.6,1){\line(1,0){1}}
\put(-2.6,0){\line(1,0){1}}

\put(-1.6,2){\line(-1,1){0.2}}
\put(-1.6,2){\line(-1,-1){0.2}}

\put(-1.6,1){\line(-1,1){0.2}}
\put(-1.6,1){\line(-1,-1){0.2}}
\put(-2.6,1){\color{green} \oval(0.14,0.14)[rt]}
\put(-2.47,1.12){\makebox(0,0)[cc]{\color{green} \tiny $ \pi $}}
\put(-2.8,1.2){\thicklines \color{green} \line(1,-1){0.4}}
\put(-2.8,0.8){\color{green} \framebox(0.4,0.4)[cc]{}}
\put(-2.4,0.6){\makebox(0,0)[lc]{\color{green} $T= 0 $}}
\put(-2.2,0.9){\color{blue} \rule{0.4cm}{0.4cm} }
\put(-2.2,1.25){\makebox(0,0)[lc]{\color{blue}$- \pi $}}

\put(-1.6,0){\line(-1,1){0.2}}
\put(-1.6,0){\line(-1,-1){0.2}}

\put(-1.2,1.7){\color{orange} \rule{0.2cm}{1.2cm} }
\put(-1.2,-0.3){\color{orange} \rule{0.2cm}{1.2cm} }

\put(1.22,1.3){\makebox(0,0)[lc]{\color{green} $T=\frac{1}{3}$}}
\put(1,1){\color{green} \oval(0.14,0.14)[rt]}
\put(1.13,1.12){\color{green} \makebox(0,0)[cc]{\tiny $\pi$}}
\put(0.8,1.2){\color{green} \line(1,-1){0.4}}
\put(0.8,0.8){\color{green} \framebox(0.4,0.4)[cc]{}}
\put(0.5,1){\line(1,0){1.00}}
\put(1,0.5){\line(0,1){1.00}}
\put(0.5,0){\line(1,0){1.00}}
\put(1,-0.5){\line(0,1){1.00}}
\put(1.5,0){\line(1,0){1.00}}
\put(2,-0.5){\line(0,1){1.00}}
\put(0,2){\line(1,0){0.50}}
\put(-0.48,2){\makebox(0,0)[cc]{$\vert \textsf{S}_z : -1 \rangle$}}
\put(3,-0.5){\line(0,-1){0.50}}
\put(3,-1){\line(1,1){0.2}}
\put(3,-1){\line(-1,1){0.2}}
\put(3,-1.2){\makebox(0,0)[cc]{$\vert \textsf{S}_x : +1 \rangle$}}
\put(0.5,2){\line(1,0){0.5}}
\put(1,2){\line(0,-1){0.5}}
\put(0,1){\line(1,0){0.50}}
\put(-0.48,1){\makebox(0,0)[cc]{$\vert \textsf{S}_z : 0 \rangle$}}
\put(2,-0.5){\line(0,-1){0.50}}
\put(2,-1){\line(1,1){0.2}}
\put(2,-1){\line(-1,1){0.2}}
\put(2,-1.2){\makebox(0,0)[cc]{$\vert \textsf{S}_x : 0 \rangle$}}
\put(1.5,1){\line(1,0){0.5}}
\put(2,1){\line(0,-1){0.5}}
\put(0,0){\line(1,0){0.50}}
\put(-0.48,0){\makebox(0,0)[cc]{$\vert \textsf{S}_z : +1 \rangle$}}
\put(1,-0.5){\line(0,-1){0.50}}
\put(1,-1){\line(1,1){0.2}}
\put(1,-1){\line(-1,1){0.2}}
\put(1,-1.2){\makebox(0,0)[cc]{$\vert \textsf{S}_x : -1 \rangle$}}
\put(2.5,0){\line(1,0){0.5}}
\put(3,0){\line(0,-1){0.5}}
\put(1.22,0.3){\makebox(0,0)[lc]{\color{green} $T=\frac{3}{4}$}}
\put(1,0){\color{green} \oval(0.14,0.14)[rt]}
\put(1.13,0.12){\color{green} \makebox(0,0)[cc]{\tiny $\pi$}}
\put(0.8,0.2){\color{green} \line(1,-1){0.4}}
\put(0.8,-0.2){\color{green} \framebox(0.4,0.4)[cc]{}}
\put(2.22,0.3){\makebox(0,0)[lc]{\color{green} $T=\frac{1}{3}$}}
\put(2,0){\color{green} \oval(0.14,0.14)[rt]}
\put(2.13,0.12){\makebox(0,0)[cc]{\color{green} \tiny $\pi$}}
\put(1.8,0.2){\color{green} \line(1,-1){0.4}}
\put(1.8,-0.2){\color{green} \framebox(0.4,0.4)[cc]{}}
\end{picture}
\end{center}
\caption{(Color online) Configuration of a random number generator
with a preparation and a measurement stage, including filters blocking
$\vert \textsf{S}_z : -1 \rangle$
and
$\vert \textsf{S}_z : +1 \rangle$.
(For ideal beam splitters, these filters would not be required.)
The measurement stage (right array) realizes a unitary quantum gate $\textsf{U}_x$, corresponding  to the projectors onto
the $\textsf{S}_x$ state observables for spin state measurements along the $x$-axis,
in terms of generalised beam splitters.0
\label{2012-incomput-proofsC-f1}}
\end{figure}
\fi

While many variations on the unitary matrix to represent a beam splitter exist~\cite{rzbb,zeilinger:882,Campos:1989fk,green-horn-zei},
without loss of generality we can represent an arbitrary $U(2)$ matrix realized by a beam splitter and external phase shift as
\begin{equation}
        \begin{pmatrix}
\sqrt{T} & i e^{i \phi}\sqrt{R}\\
i\sqrt{R} & e^{i \phi}\sqrt{T}
\end{pmatrix},
\end{equation}
where $\phi$ represents the phase of an external phase shifter on the second input port, and $T,R\in [0,1]$ are the transmittance and reflectance of the beam splitter respectively (with $R+T=1$).
The beam splitter arrangement to realize $\textsf{U}_x$
can be found by transforming $\textsf{U}_x$ into the identity matrix $I_3$ by successive right-multiplication by adjoints of $U(2)$ matrices of the above form---each one making an individual off-diagonal element equal to zero---followed by a final set of phase shifters~\cite{rzbb}.

In our specific case, we have
\begin{equation}
        \begin{pmatrix}
        1 & 0 & 0\\
        0 & -i & 0\\
        0 & 0 & -i
        \end{pmatrix}
        \cdot
        \begin{pmatrix}
        \sqrt{\frac{1}{3}} & \sqrt{\frac{2}{3}} & 0\\
        i\sqrt{\frac{2}{3}} & -i\sqrt{\frac{1}{3}} & 0\\
        0 & 0 & 1
        \end{pmatrix}
        \cdot
        \begin{pmatrix}
        \sqrt{\frac{3}{4}} & 0 & -i\sqrt{\frac{1}{3}}\\
        0 & 1 & 0\\
        i\sqrt{\frac{1}{4}} & 0 & -\sqrt{\frac{3}{4}}
        \end{pmatrix}
        \cdot
        \begin{pmatrix}
        1 & 0 & 0\\
        0 & \sqrt{\frac{1}{3}} & \sqrt{\frac{2}{3}}\\
        0 & i\sqrt{\frac{2}{3}} & -i\sqrt{\frac{1}{3}}\\
        \end{pmatrix}
        =\textsf{U}_x.
\end{equation}
This corresponds to three beam splitters with transmittances $T_{3,2}=T_{2,1}=\frac{1}{3}$, $T_{3,1}=\frac{3}{4}$, and phases $\phi_{3,2}=\phi_{2,1}=-\pi/2$, $\phi_{3,1}=\pi \raisebox{.8mm}{,}$
where $T_{i,j}$ and $\phi_{i,j}$ are the parameters for the beam splitter operating on beams $i$ and $j$ (beams 1,2,3 correspond to $\textsf{S}_z=+1,0,-1$ respectively).
Two final phase shifts of $-\pi/2$ are needed on beams 2 and 3.
The physical realization of $\textsf{U}_x$ is depicted in Fig.~\ref{2012-incomput-proofsC-f1}.

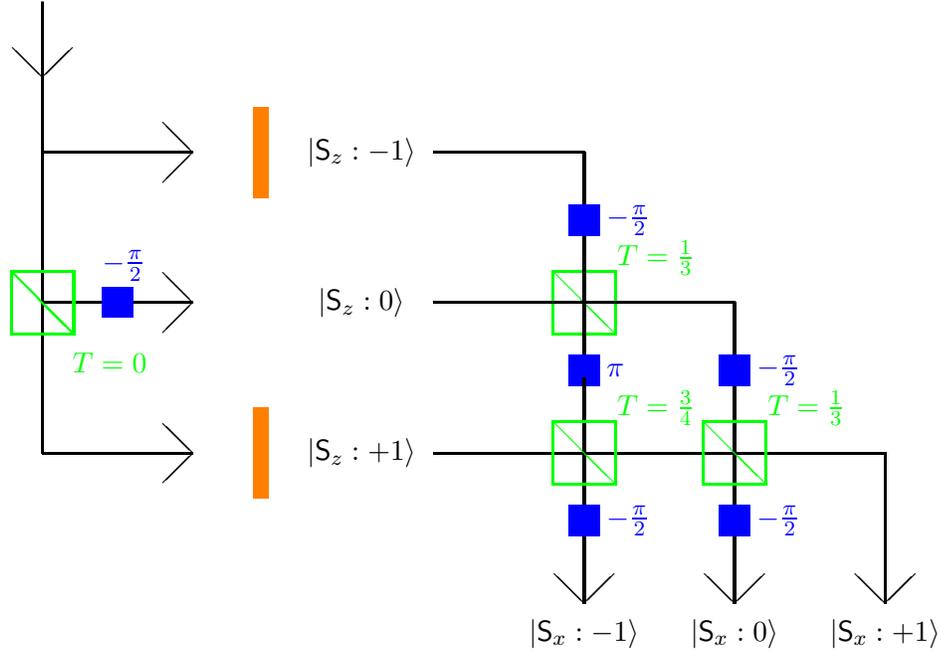
\begin{figure}[ht]
\begin{center}

\unitlength 20mm
\linethickness{0.8pt}
\begin{picture}(6,5)(-3.0,-1.5)

\put(-2.6,3){\line(0,-1){3}}

\put(-2.6,2.5){\line(1,1){0.2}}
\put(-2.6,2.5){\line(-1,1){0.2}}

\put(-2.6,2){\line(1,0){1}}
\put(-2.6,1){\line(1,0){1}}
\put(-2.6,0){\line(1,0){1}}

\put(-1.6,2){\line(-1,1){0.2}}
\put(-1.6,2){\line(-1,-1){0.2}}

\put(-1.6,1){\line(-1,1){0.2}}
\put(-1.6,1){\line(-1,-1){0.2}}
\put(-2.8,1.2){\thicklines \color{green} \line(1,-1){0.4}}
\put(-2.8,0.8){\color{green} \framebox(0.4,0.4)[cc]{}}
\put(-2.4,0.6){\makebox(0,0)[lc]{\color{green} $T= 0 $}}
\put(-2.2,0.9){\color{blue} \rule{0.4cm}{0.4cm} }
 \put(-2.2,1.25){\makebox(0,0)[lc]{\color{blue}$-\frac{\pi}{2}$}}

\put(-1.6,0){\line(-1,1){0.2}}
\put(-1.6,0){\line(-1,-1){0.2}}

\put(-1.2,1.7){\color{orange} \rule{0.2cm}{1.2cm} }
\put(-1.2,-0.3){\color{orange} \rule{0.2cm}{1.2cm} }

\put(1.22,1.3){\makebox(0,0)[lc]{\color{green} $T=\frac{1}{3}$}}
\put(0.8,1.2){\color{green} \line(1,-1){0.4}}
\put(0.8,0.8){\color{green} \framebox(0.4,0.4)[cc]{}}
\put(0.5,1){\line(1,0){1.00}}
\put(1,0.5){\line(0,1){1.00}}
\put(0.9,0.45){\color{blue} \rule{0.4cm}{0.4cm} }
\put(1.15,0.55){\makebox(0,0)[lc]{\color{blue}$\pi$}}
\put(0.5,0){\line(1,0){1.00}}
\put(1,-0.5){\line(0,1){1.00}}
\put(1.5,0){\line(1,0){1.00}}
\put(2,-0.5){\line(0,1){1.00}}
\put(0,2){\line(1,0){0.50}}
\put(-0.48,2){\makebox(0,0)[cc]{$\vert \textsf{S}_z : -1 \rangle$}}
\put(3,-0.5){\line(0,-1){0.50}}
\put(3,-1){\line(1,1){0.2}}
\put(3,-1){\line(-1,1){0.2}}
\put(3,-1.2){\makebox(0,0)[cc]{$\vert \textsf{S}_x : +1 \rangle$}}
\put(0.5,2){\line(1,0){0.5}}
\put(1,2){\line(0,-1){0.5}}
\put(0,1){\line(1,0){0.50}}
\put(-0.48,1){\makebox(0,0)[cc]{$\vert \textsf{S}_z : 0 \rangle$}}
\put(2,-0.5){\line(0,-1){0.50}}
\put(2,-1){\line(1,1){0.2}}
\put(2,-1){\line(-1,1){0.2}}
\put(2,-1.2){\makebox(0,0)[cc]{$\vert \textsf{S}_x : 0 \rangle$}}
\put(1.5,1){\line(1,0){0.5}}
\put(2,1){\line(0,-1){0.5}}
\put(0,0){\line(1,0){0.50}}
\put(-0.48,0){\makebox(0,0)[cc]{$\vert \textsf{S}_z : +1 \rangle$}}
\put(1,-0.5){\line(0,-1){0.50}}
\put(1,-1){\line(1,1){0.2}}
\put(1,-1){\line(-1,1){0.2}}
\put(1,-1.2){\makebox(0,0)[cc]{$\vert \textsf{S}_x : -1 \rangle$}}
\put(2.5,0){\line(1,0){0.5}}
\put(3,0){\line(0,-1){0.5}}
\put(1.22,0.3){\makebox(0,0)[lc]{\color{green} $T=\frac{3}{4}$}}
\put(0.8,0.2){\color{green} \line(1,-1){0.4}}
\put(0.8,-0.2){\color{green} \framebox(0.4,0.4)[cc]{}}
\put(2.22,0.3){\makebox(0,0)[lc]{\color{green} $T=\frac{1}{3}$}}
\put(1.8,0.2){\color{green} \line(1,-1){0.4}}
\put(1.8,-0.2){\color{green} \framebox(0.4,0.4)[cc]{}}
 \put(1.9,0.45){\color{blue} \rule{0.4cm}{0.4cm} }
 \put(2.15,0.55){\makebox(0,0)[lc]{\color{blue}$ -\frac{\pi}{2}$}}
 \put(0.9,1.45){\color{blue} \rule{0.4cm}{0.4cm} }
\put(1.15,1.55){\makebox(0,0)[lc]{\color{blue}$-\frac{\pi}{2}$}}
 \put(0.9,-0.55){\color{blue} \rule{0.4cm}{0.4cm} }
 \put(1.15,-0.45){\makebox(0,0)[lc]{\color{blue}$-\frac{\pi}{2}$}}
 \put(1.9,-0.55){\color{blue} \rule{0.4cm}{0.4cm} }
\put(2.15,-0.45){\makebox(0,0)[lc]{\color{blue}$-\frac{\pi}{2}$}}
\end{picture}
\end{center}
\caption{(Color online) Configuration of a random number generator
with a preparation and a measurement stage, including filters blocking
$\vert \textsf{S}_z : -1 \rangle$
and
$\vert \textsf{S}_z : +1 \rangle$.
(For ideal beam  splitters, these filters would not be required.)
The measurement stage (right array) realizes a unitary quantum gate $\textsf{U}_x$, corresponding  to the projectors onto
the $\textsf{S}_x$ state observables for spin state measurements along the $x$-axis,
in terms of generalised beam splitters.
\label{2012-incomput-proofsC-f1}}
\end{figure}

This setup is equivalent to the spin-1 setup for which we are guaranteed value indefiniteness under the conditions discussed in the previous section.
Even in the case of non-perfectly configured beam splitters, as long as the observable corresponding to the unitary transformation implemented by the beam splitters has eigenstates $\ket{a=\pm 1}$ (corresponding to output ports 1 and 3) which fall within the bounds $\sqrt{\frac{5}{14}} \le \iprod{\textsf{S}_z=0}{a=\pm 1}\le \frac{3}{\sqrt{14}}$ then the QRNG will still be protected by value indefiniteness.
As discussed in the previous section, this allows for a considerable amount of error (more than would be desirable with respect to deviation from 50/50 bias) under which value indefiniteness is still guaranteed.

\if01
An inspection of the generalised beam splitter setup reveals that it actually consists of a standard 50:50 beam splitter device
used for quantum coin tosses
\cite{svozil-qct,rarity-94,zeilinger:qct,stefanov-2000,0256-307X-21-10-027,wang:056107,fiorentino:032334,svozil-2009-howto,10.1038/nature09008},
with the addition of additional optical components remaining ``idle'' or ``dormant''
in the case of ideal components.
This is due to the effectively quasi two-dimensionality of the dichotomic outcome setup,
which may be not so evident in the case of spin state measurements on spin one particles.
However, one should keep in mind that
the device is capable of performing transformations and rendering observables in full three dimensional Hilbert space,
with all the counterfactual potentialities present.
This in a sense, is not dissimilar to counterfactual computation  \cite{elitzur-vaidman:1}.

{\color{blue}
A more general configuration that involves experimental robustness discussed earlier in section \ref{2012-incomput-proofsCJ-er}
can be obtained by considering
\begin{equation}
\begin{array}{l}
 \ket{\textsf{S}(\theta ,0) : +1} \equiv     \left( \cos ^2 \frac{\theta }{2} , \frac{\sin   \theta   }{\sqrt{2}} ,\sin^2 \frac{\theta }{2} \right),\\
 \ket{\textsf{S}(\theta ,0) : 0}  \equiv     \left( \frac{\sin  \theta}{\sqrt{2}} ,-   \cos  \theta    ,-\frac{\sin  \theta }{\sqrt{2}} \right),\\
 \ket{\textsf{S}(\theta ,0) : -1} \equiv       \left( \sin ^2 \frac{\theta }{2} ,-\frac{\sin   \theta    }{\sqrt{2}} ,\cos^2 \frac{\theta }{2} \right),
\end{array}
\end{equation}
thereby forming a unitary matrix $\textsf{U}_x $
which correspond to the spin state operator $\textsf{S}(\theta ,0)$ for spin state measurements along $\theta$;
more explicitly,
\begin{equation}
\textsf{U}(\theta ,0)
=
\left(
\begin{array}{cccc}
\cos ^2 \frac{\theta }{2} &\frac{\sin   \theta   }{\sqrt{2}} &\sin^2 \frac{\theta }{2} \\
\frac{ \sin  \theta}{\sqrt{2}}  &-   \cos  \theta    &-\frac{\sin  \theta  }{\sqrt{2}}  \\
 \sin ^2 \frac{\theta }{2} &-\frac{\sin   \theta   }{\sqrt{2}}  &\cos^2 \frac{\theta }{2}
\end{array}
\right).
\end{equation}

}
\fi
\section{Monitoring Value Indefiniteness}

The rendition of value indefiniteness requires a quantised system with at least three mutually exclusive outcomes,
corresponding to an associated Hilbert space dimension equal to the number of these outcomes---a direct consequence of the Kochen-Specker theorem.

Of course, if one is willing to accept physical value indefiniteness
based purely on formal  Hilbert space models of quantum mechanics \cite{v-neumann-55},
there is no further need of empirical evidence.
In this line of thinking, Theorem~\ref{thm:twonotvaluedefinite}, and hence the quantum value indefiniteness resulting from it via Corollary~\ref{cor:twonotvaluedefinite}, needs no more empirical corroboration
than the arithmetic fact that, in Peano arithmetic with standard addition, one plus one equals two.



\if01
{\color{red}
There are some indirectly obtained physical remedies which are often referred to as
``proofs of the Kochen-Specker theorem,''
or to ``proofs of quantum contextuality''
\cite{hasegawa:230401,Bartosik-09,kirch-09,PhysRevLett.103.160405,Lapkiewicz-11}.
Thus if one is
willing to either accept as proofs certain Bell type arguments about a series of successive (``one-after-another'') outcomes \cite{PhysRevLett.103.160405},
or to accept the type of nonclassical outcomes typically encountered in
empirical realizations of Greenberger-Horne-Zeilinger type arguments \cite{PhysRevLett.82.1345,panbdwz},
then indeed value indefiniteness or contextuality are physically provable.

With this {\it caveat} we propose to {\em monitor} all protocols invoking value indefiniteness or contextuality by
physical experiments testing the nonclassical nature of the quantised systems involved.
This could be either achieved by invoking the aforementioned violations of Bell-type inequalities \cite{10.1038/nature09008}, or by
experiments invoking  Greenberger-Horne-Zeilinger type measurements.
The latter ones may be preferable, since, at least ideally, they do not involve any statistics,
but require a violation of local realism at every triple of outcomes.
Of course, in order to monitor this nonclassical behaviour, protocols must be adapted to
include such measurements involving multipartite correlations.

In the concrete beam splitter configuration mentioned earlier one could monitor the nonclassicality of the setup
by first producing the (only \cite{schimpf-svozil}) singlet state
$\frac{1}{{\sqrt{3}}}\left(-|0,0\rangle+|-1,1\rangle+|1,-1\rangle\right)$,
and then subjecting the two entangled quanta to beam splitter analysis.
We refer to Refs.~\cite{zukowski-97} and \cite{svozil-2004-analog} for two principal approaches.
The Bell-type inequalities violated in this case are the two particle  three properties per particle
inequalities derived in Refs.~\cite{2000-poly,collins-gisin-2003,sliwa-2003}.
}
\fi

QRNGs which monitor Bell-inequality violation simultaneously with
bit-generation have been proposed in the literature~\cite{10.1038/nature09008,Vazirani28072012}.
Given the non-trivial assumptions used in the proof of  Theorem~\ref{thm:twonotvaluedefinite}---in particular, the mutual physical coexistence of
complementary observables---should our QRNG be monitored in this way too, in addition to  value indefiniteness certification?

First,  we stress that, in contrast with our proposed QRNG,  the aforementioned devices require an initial random seed and hence
operate as a secure randomness {\em expander}, rather than {\em generator}: The quality
of randomness produced by such a device depends crucially upon the quality
of randomness of the seed.

Second, violation of Bell-inequalities alone is a purely statistical phenomenon and only
indicates non-classical correlations: in no way does it necessitate a Hilbert-space
structure and hence it cannot give the certification of (strong) incomputability that
our proposal does via value indefiniteness.

Third, in the case that our QRNG is treated as an untrusted-device, as is common in
cryptography (due to the users inability to verify the device's workings), the set up
could be modified to test such inequalities. This is the scenario in which monitoring
inequality violation has most to offer, since violation of Bell-inequalities
can be derived from Kochen-Specker type arguments \footnote{Such violations are
often referred  to as ``proofs of the Kochen-Specker theorem,'' or ``proofs of quantum contextuality''\cite{hasegawa:230401,Bartosik-09,kirch-09,PhysRevLett.103.160405,Lapkiewicz-11}.}
and thus gives some indication of non-classicality in the absence of trust in the device,
even if it cannot guarantee incomputability. An even better monitoring method---which might necessitate a revision of our current QRNG
set up---may use the type of
non-classical outcomes typically encountered in empirical realizations of Greenberger-Horne-Zeilinger
type arguments \cite{PhysRevLett.82.1345,panbdwz} because, at least ideally, they do not involve any statistics,
but require a violation of local realism  at every triple of outcomes.

To summarize, we have presented
a formal conceptualization of {\em value (in-)definiteness,}
and proven that there always exists an admissible assignment function making a \emph{single} observable value definite;
one cannot hope to prove \emph{all} observables are value indefinite.
We also showed that, in an extension of the Kochen-Specker theorem,
after preparing a pure state in three-dimensional Hilbert space,
certain precisely identified observables are {\em provably value indefinite.}

We have applied these results to a proposal to generate bit sequences by
a quantum random number generator. Any such sequence is, as we showed, then ``certified by'' quantum value indefiniteness
(in the sense of the Bell-, Greenberger-Horne-Zeilinger-, and Kochen-Specker theorems) to produce a strongly incomputable sequence of bits.

To what extent we can guarantee value indefiniteness remains an open question.
We know that not all observables can be value indefinite,
and at least those in the star-shaped setup of Fig.~\ref{2012-incomputability-proofs-f2}
can be guaranteed to be, but how far does this value indefiniteness go?
We conjecture that this is as far as one can go;
that {\em only a single} observable in the Hilbert space can be assigned  the value one,
and only those orthogonal to the said observable can be assigned the value 0---any other observables must,
 under the assumption of non-contextuality be value indefinite.
\section*{Acknowledgements}
We are grateful to Kohtaro Tadaki for insightful comments which improved the paper, as well as the anonymous referees who provided helpful comments.
We thank Michael Reck for the code producing the generalised beam splitter setup for an arbitrary unitary transformation.
Abbott, Calude and Svozil have been supported in part by Marie Curie FP7-PEOPLE-2010-IRSES Grant RANPHYS.
Calude's contribution was done in part during his tenure as Visiting Fellow of the Isaac Newton Institute for Mathematical Sciences (June--July 2012). Conder has been supported in part by a University of Auckland Summer Scholarship (2012).
Svozil's contribution was done in part during his visiting honorary appointment at the University of Auckland (February--March 2012),
and a visiting professorship at the University of Cagliari (May--July 2012).


%

\end{document}